\documentclass[article,ij4uq]{ij4uq}              

\usepackage{comment}

\usepackage[hang]{footmisc}
\renewcommand{\dmyy}{21}
\renewcommand{\myyear}{2022}
\renewcommand{\today}{}
\SetKwSty{<text>}

\newtheorem{example}{Example}
\renewcommand{\vec}[1] {\ensuremath{\boldsymbol{#1}}}

\def\b{{\vec b}}
\def\z{{\vec z}}
\def\t{{\vec \theta}}
\def\r{{\mathbb R}}
\def\d{{\vec d}}
\def\e{{\vec e}}

\def\g{{\vec g}}

\def\v{{\vec v}}
\def\s{{\vec s}}
\def\m{{\vec m}}

\def\zero{{\vec 0}}
\def\mgamma{{\vec \Gamma}}

\def\x{{\vec x}}

\def\real{\mathbb{R}}

\def\forcing{h_{\text{flux}}}

\def\Bop{\boldsymbol{\mathcal{B}}}

\def\Hop{\boldsymbol{\mathcal{H}}}
\def\Pop{\boldsymbol{\mathcal{P}}}
\def\Vop{\boldsymbol{\mathcal{V}}}

\def\mc{{\vec C}}

\def\mh{{\vec H}}
\def\mi{{\vec I}}
\def\my{{\vec Y}}

\def\ms{{\vec S}}
\def\md{{\vec D}}
\def\mt{{\vec T}}

\def\mlambda{{\vec \Lambda}}
\def\momega{{\vec \Omega}}
\def\mq{{\vec Q}}
\def\pN{\mathcal{N}}

\newcommand{\V}[1]{\boldsymbol{#1}}
\newcommand{\M}[1]{\boldsymbol{#1}}
\newcommand{\Mc}[1]{\boldsymbol{\mathcal{#1}}}
\newcommand{\thetabar}{\overline{\vec{\theta}}}
\newcommand{\bmat}[1]{\begin{bmatrix}#1 \end{bmatrix}}

\newcommand{\joeynote}[1]{{\color{red}Joey : {#1}}}

\newcommand{\willnote}[1]{{\color{orange}William: {#1}}}

\begin{document}

\volume{Volume x, Issue x, \myyear\today}
\title{Hyper-differential sensitivity analysis in the context of Bayesian inference applied to ice-sheet problems}
\titlehead{Post-Optimality in Land-Ice}
\authorhead{William Reese, Joseph Hart, Bart van Bloemen Waanders,Mauro Perego, John Jakeman, and Arvind Saibaba}

\author[1]{William Reese}
\author[2]{Joseph Hart}
\corrauthor[3]{Bart van Bloemen Waanders}
\author[4]{Mauro Perego}
\author[5]{John Jakeman}
\author[6]{Arvind K.\ Saibaba}

\corremail{bartv@sandia.gov}
\corraddress{Sandia National Laboratories, P.O. Box 5800, Albuquerque, NM 87185}

\address[1 6]{Department of Mathematics
Box 8205, NC State University
Raleigh, NC 27695-8205}

\address[2 3 4 5]{Optimization and Uncertainty Quantification, Sandia National Laboratories, P.O. Box 5800, Albuquerque, NM 87123-1320}

\date{}

\abstract{Inverse problems constrained by partial differential equations (PDEs) play a critical role in model development and calibration. In many applications, there are multiple uncertain parameters in a model which must be estimated. Although the Bayesian formulation is attractive for such problems, computational cost and high dimensionality frequently prohibit a thorough exploration of the parametric uncertainty. A common approach is to reduce the dimension by fixing some parameters  (which we will call auxiliary parameters) to a best estimate and use techniques from PDE-constrained optimization to approximate properties of the Bayesian posterior distribution. For instance, the maximum a posteriori probability (MAP) and the Laplace approximation of the posterior covariance can be computed. In this article, we propose using hyper-differential sensitivity analysis (HDSA) to assess the sensitivity of the MAP point to changes in the auxiliary parameters. We establish an interpretation of HDSA as correlations in the posterior distribution. Our proposed framework is demonstrated on the inversion of bedrock topography for the Greenland ice sheet with uncertainties arising from the basal friction coefficient and climate forcing (ice accumulation rate).
}

\keywords{Hyper-differential sensitivity analysis, Bayesian inverse problems, Inversion}

\maketitle


\section{Introduction} \label{sec:intro}
Large scale inverse problems occur in a range of geoscience applications from seismicity to ice sheet flows \cite{Ghattas_2013,Ghattas_2014}.  In such problems, the quantities of interest typically depend on unknown parameters that describe material properties, source terms, boundary and initial conditions in the governing partial differential equations (PDEs). The goal is to reconcile the differences between measurements and numerical predictions by estimating or reconstructing the unknown parameters. 
This is fraught with many challenges such as limited availability and noise in the observed data, the need to suitably regularize the problem, and additional uncertainties present in the model. 

We discuss some of these challenges in the context of an ice sheet model, the driving application for this work. Ice sheets play an important role in the global climate through their effects
on sea level rise. Sea level rise may potentially cause severe flooding and weather changes that will negatively impact wildlife, agriculture, and coastal infrastructure~\cite{ipcc_sea_level_rise_2019}. We consider an ice sheet model with several uncertain parameters: the bedrock topography beneath the ice, accumulation and ablation on the upper surface of the ice, and basal friction on the bottom of the ice sheet. The bedrock topography is estimated using data acquired from aircraft flyovers but is uncertain because of data sparsity (the paths flown by aircraft) \cite{van_pelt_2013,bonan_2014}. Accumulation and ablation, a forcing term in the model, requires climate information to be accurately specified \cite{charbit_2013,quiquet_2012}. The basal friction, a combination of several physical phenomena, is uncertain due to the limited availability of observations beneath the ice sheet.



Joint inversion~\cite{Crestel_2018} is one approach to estimate multiple uncertain parameters that may be interdependent. An example is inversion of electromagnetic and seismic parameters~\cite{Abubakar_2012}. However, joint inversion is challenging because the joint parameter dimensions can be high and it typically requires a joint regularization term that reflects spatial correlations and imposes regularity~\cite{Crestel_2018}. Some examples of joint regularization are cross gradient regularization term~\cite{Gallardo_2013} and color total variation (TV)~\cite{Blomgren_1998}.

Bayesian theory provides a suitable formalism to solve inverse problems with statistical characterizations. In practice, joint Bayesian inversion of \emph{all} unknown parameters is computationally intractable for the applications of interest. A commonly used alternative, which we also adopt in this paper, is to invert for one set of parameters (which we call {\em inversion parameters}) and fix all remaining parameters (which we call {\em auxiliary parameters}) to a nominal value. The challenge with this approach is that inferences made about the inversion parameter are relative to the value of the nominal parameters. This implies that if the nominal values change the inversion may yield very different results. Therefore, it remains to quantify the influence or sensitivity of the solution of the inverse problem's auxiliary parameters. To this end, we use hyper-differential sensitivity analysis (HDSA) to compute post-optimality sensitivities of the maximum a posteriori estimate with respect to perturbations of the auxiliary parameters.

Our approach is complementary to the Bayesian approximation error (BAE) method~\cite{kaipio_2013,babaniyi_2021,hartland_2021}. In BAE, the auxiliary parameter are marginalized out of the Bayesian posterior giving an inverse problem to estimate the inversion parameters with a cognizance of the auxiliary parameter uncertainty. In contrast, we show how HDSA provides information about the joint posterior in a neighborhood of the nominal auxiliary parameters thus elucidating dependencies between the parameters. There is also potential to combine BAE and HDSA, but such analysis is beyond the scope of this article. 

\textbf{Contributions}: We build on previous efforts~\cite{hdsa_ill_posed_inv_prob,sunseri_hdsa,HDSA} and apply HDSA to analyze the sensitivity of the estimated bedrock topography with respect to perturbations of the accumulation/ablation forcing and basal friction for a realistic model of the Greenland ice sheet. Furthermore, we introduce a new interpretation of the sensitivities in terms of correlations in the joint Bayesian posterior distribution. This links properties of the computationally intractable joint inversion problem with post-optimality sensitivities which may be computed efficiently. Additionally, we demonstrate that computing HDSA on the likelihood informed subspace (LIS), as introduced in~\cite{hdsa_ill_posed_inv_prob}, provides the eigenvectors needed to compute approximation posterior samples with the Laplace approximation~\cite{isaac_2015,Ghattas_2014}. With this novel observation, we show how post-optimality sensitivities and approximate posterior samples may be computed simultaneously. The resulting approach is applicable to large-scale nonlinear inverse problems with expensive forward models involving systems of PDEs and multiple unknown parameters.

\textbf{Overview}: In Section~\ref{sec:background}, we provide a brief overview of Bayesian inversion in the context of PDE constrained optimization and HDSA. An interpretation of HDSA as correlations in the joint posterior is established in Section~\ref{sec:bayes_interp}. Section~\ref{sec:lis} demonstrates how sensitivities and approximate posterior provides may be simultaneously computed using the LIS. The computational costs is analyzed in Section~\ref{sec:algorithms}. In Section~\ref{sec:num_results}, we demonstrate this process on an inverse problem for the bedrock topography in a region of Greenland. Concluding remarks are made in Section~\ref{sec:conclusion}.

\section{Background} \label{sec:background}
We consider inverse problems to estimate parameters $z \in Z$, where $Z$ may be finite or infinite dimensional.  Assume that $z$ can not be observed directly, but rather we have sparse and noisy observations of a state $u \in U$ ($U$ is an appropriate function space) which is related to $z$ by a PDE, $c$, in the form
\begin{equation} \label{eqn:pde}
c(u,z,\theta)=0
\end{equation}
where $\theta \in \Theta$ are uncertain parameters, referred to as auxiliary parameters, needed to define the PDE. 

Assume that the PDE is uniquely solvable for each $z \in Z$ and $\theta \in \Theta$ and let $\Psi:Z \times \Theta \to U$ denote the PDE solution operator,
i.e. 
\begin{eqnarray*}
c(\Psi(z,\theta),z,\theta)=0 \qquad \text{ for all } z \in Z, \theta \in \Theta.
\end{eqnarray*}
Let $\mathcal{F}: Z \times \Theta \to \mathbb R^d$ denote the parameter-to-observable map. Specifically,
$$
\mathcal{F}(z,\theta) = \mathcal{O} \circ \Psi(z,\theta)
$$
maps parameters $z$ and $\theta$ to observations of the PDE solution at $d$ locations via the observation operator $\mathcal{O}: U \to \mathbb R^d$.

We focus on the finite dimensional inverse problem which arises from the discretization of the PDE. To simplify the exposition, assume that $Z_h \subset Z$ and $\Theta_h \subset \Theta$ are finite dimensional (a result of the PDE discretization) with bases $\{y_1,y_2,\dots,y_m\} \subset Z_h$ and $\{\phi_1,\phi_2,\dots,\phi_n\} \subset \Theta_h$. We will use $\vec{z}=(z_1,z_2,\dots,z_m)^T \in \mathbb R^m$ and $\vec{\theta}=(\theta_1,\theta_2,\dots,\theta_n)^T \in \mathbb R^n$ to denote coordinate representations of elements in $Z_h$ and $\Theta_h$ (which may be function spaces). Let $\vec{f}:\mathbb R^m \times \mathbb R^n \to \mathbb R^d$ denote the discretized parameter-to-observable map. We assume that the data $\vec{d} \in \mathbb R^d$ is related to parameters $\z$ and $\vec{\theta}$ as
\begin{equation}\label{eq:data_eqn}
\vec{d} = \vec{f}(\vec{z},\vec\theta) + \vec\eta,
\end{equation}
where $\vec\eta \in \mathbb R^d$ is noise. The goal of the finite dimensional inverse problem is to estimate $\vec{z}$  and $\vec{\theta}$ from $\vec{d}$. 

\subsection{Bayesian inverse problems} \label{ssec:bayes_inv_probs}
Inverse problems are frequently ill-posed in the sense that there are many different $\vec{z}$'s and $\vec{\theta}$'s such that $\V{f}(\vec{z},\vec{\theta}) \approx \vec{d}$. This ill-posedness motivates a Bayesian approach to the inverse problem.  We review core concepts below and direct the reader to~\cite{Stuart_2010,Kaipio_2005} for a detailed introduction to Bayesian inverse problems.

Assume that the data $\vec{d}$ is related to the parameters $\vec{z}$ and $\vec\theta$ as in \eqref{eq:data_eqn} where $\vec\eta \sim \mathcal{N}(\vec{0},\M\Gamma_{\text{noise}})$. Throughout, we also assume that the prior distributions of $\vec{z}$ and $\vec\theta$ are independent Gaussians $\vec{z}\sim \mathcal{N}(\vec{\mu}_{\vec{z}},\M\Gamma_{\vec{z}})$ and  $\vec{\theta}\sim \mathcal{N}(\vec{\mu}_{\vec{\theta}},\M\Gamma_{\vec{\theta}})$, where both $\M\Gamma_{\V{z}}$ and $\M\Gamma_{\V\theta}$ are symmetric positive definite. Applying Bayes rule, the posterior distribution is
\begin{equation}\label{eqn:joint} \pi_{\rm joint}(\vec{z},\vec{\theta}| \vec{d}) =\frac{1}{C_{\rm joint}} \exp \left( -\frac{1}{2} \|\V{f}(\vec{z},\vec{\theta})-\vec{d}\|_{ \M\Gamma_{\text{noise}}^{-1}}^2   -\frac12 \|\vec{z}-\vec\mu_{\vec{z}} \|_{\M\Gamma_{\vec{z}}^{-1}}^2 -\frac12 \| \vec{\theta}-\vec\mu_{\vec{\theta}}\|_{\M\Gamma_{\vec{\theta}}^{-1}}^2  \right), \end{equation}
where $C_{\rm joint}$ denotes a normalizing constant which is difficult to compute and is unimportant for this discussion. 
However, jointly estimating $\vec{z}$ and $\vec\theta$ is challenging since the dimensions $m$ and $n$ may both be large and the data may not be sufficiently informative of both parameters. Therefore, one approach in practice is to fix the auxiliary parameters to a nominal value, denoted by $\thetabar$, and estimate $\vec{z}$ from the data. 

From a probabilistic perspective, fixing $\vec\theta$ to a nominal value corresponds to conditioning the joint posterior $\pi_{\rm joint}$ on $\vec\theta=\thetabar$. In other words, analyzing the conditional posterior
\begin{equation}
\label{eqn:cond}
 \pi_{\rm cond}(\vec{z}| \thetabar,\vec{d}) = \frac{1}{C_{\rm cond}(\thetabar)} \exp \left( -\frac{1}{2} \|\V{f}(\vec{z},\thetabar)-\vec{d}\|_{ \M\Gamma_{\text{noise}}^{-1}}^2   -\frac12 \|\vec{z}-\vec\mu_{\vec{z}} \|_{\M\Gamma_{\vec{z}}^{-1}}^2  \right), 
 \end{equation}
where the normalizing constant $C_{\rm cond}(\thetabar)$ is a function of the nominal value of the auxiliary parameters. Computing $C_{\rm cond}(\thetabar)$ is challenging as it requires integrating $\pi_{\rm joint}(\vec{z},\vec{\theta}| \vec{d})$ with respect to $\vec{z}$; however, it is unimportant for this discussion. 

Observe that the maximum a posteriori probability (MAP) point of $\pi_{\rm cond}(\vec{z}| \thetabar,\vec{d})$ is obtained by solving the optimization problem
\begin{align}\label{inv_prob_RS}
& \min\limits_{\vec{z} \in \mathbb R^m} 
J(\vec{z},\V\thetabar):=M(\vec{z},\vec{\thetabar}) + R(\vec{z})
\end{align}
where 
\begin{equation}
\label{eqn:M_and_R}
M(\vec{z},\V\theta) := \frac{1}{2} \|\V{f}(\vec{z},\vec{\theta})-\vec{d}\|_{ \M\Gamma_{\text{noise}}^{-1}}^2 \qquad \text{and}  \qquad R(\vec{z}) :=  \frac{1}{2} \|\vec{z}-\V{\mu}_{\V{z}}\|_{ \M\Gamma_{\V{z}}^{-1}}^2 
\end{equation}
are the negative log-likelihood and the negative log of $\vec{z}$'s prior PDF (with normalizing constants omitted). In the context of the optimization problem~\eqref{inv_prob_RS}, we will also refer to $M$ and $R$ as the misfit and regularization, respectively. 

For high dimensional inverse problems constrained by computationally intensive PDEs, sampling from the posterior distribution is challenging. When $\V{f}$ is a linear function, the posterior distribution is Gaussian and its covariance matrix is given by the inverse Hessian of the negative log of the posterior PDF. 
For a general nonlinear $\V{f}$, the posterior covariance is non-Gaussian and significantly more challenging to estimate. 
Computing the MAP point~\eqref{inv_prob_RS} is computationally intensive since it requires many PDE solves \cite{Ghattas_2013,Ghattas_2014,Biegler_11}. However, techniques from PDE-constrained optimization may be leveraged to solve~\eqref{inv_prob_RS} at large scales. We use techniques including finite element discretization, matrix free linear algebra, adjoint-based derivative computation, and parallel computing. The reader is referred to \cite{Vogel_99, Archer_01,Haber_01,Vogel_02,Biegler_03,Biros_05,Laird_05,Hintermuller_05,Hazra_06,Biegler_07,Borzi_07,Hinze_09,Biegler_11,Kouri2018} for a sampling of the PDE-constrained optimization literature. From the perspective of the Bayesian inverse problem,  PDE-constrained optimization is a valuable tool to efficiently compute the MAP point (and possibly approximate the covariance). 

\subsection{Hyper-differential sensitivity analysis}
\label{ssec:hdsa_intro}
Post-optimality sensitivity analysis is predicated on employing the Implicit Function Theorem to write the solution of an inverse problem as a function of the auxiliary parameters. This function and its Jacobian are defined within a neighborhood of the nominal value of the auxiliary parameters. Through a combination of tools from post-optimality sensitivity analysis, PDE-constrained optimization, and numerical linear algebra, HDSA provides unique and valuable insights for optimal control and deterministic inverse problems \cite{HDSA,sunseri_hdsa,hdsa_gsvd}.
This subsection provides essential background to prepare for our extension of HDSA to Bayesian inverse problems. To facilitate our analysis, assume that the objective function $J:\mathbb R^m \times \mathbb R^n \to \mathbb R$ in (\ref{inv_prob_RS}) is twice continuously differentiable with respect to $(\vec{z},\vec{\theta})$. 

Let $\vec{z}^\star$ be a local minimum of the objective function with auxiliary parameters fixed to $\vec{\theta}=\thetabar\in \mathbb R^n$, i.e. it is a MAP point of the posterior for $\vec{z}$ when conditioning on $\vec{\theta}=\thetabar$. A fundamental assumption of post-optimality sensitivity analysis is that $\vec{z}^\star$ satisfies the well-known first and second order optimality conditions:
\begin{enumerate}[label=(A\textnormal{\arabic*})]
\item $\nabla_{\vec{z}} J(\vec{z}^\star,\thetabar)=0$, \label{A1} 
\item $\nabla_{\vec{z},\vec{z}} J(\vec{z}^\star, \thetabar)$ is positive definite, \label{A2}
\end{enumerate}
where $\nabla_{\vec{z}} J$ and $\nabla_{\vec{z},\vec{z}} J$ denote the gradient and Hessian of $J$ with respect to $\vec{z}$, respectively. Then the Implicit Function Theorem gives the existence of a continuously differentiable operator $\Mc{G}:N(\thetabar) \to \mathbb R^m$, defined on a neighborhood of $\thetabar$, such that $\nabla_{\vec{z}} J(\Mc{G}(\vec{\theta}); \vec{\theta})=0$ for all $\vec{\theta} \in N(\thetabar)$. Furthermore, the Jacobian of $\Mc{G}$, evaluated at $\thetabar$, is given by
\begin{eqnarray}
\label{opt_sol_sensitivity}
 \Mc{G}'(\thetabar) = -\Mc{H}^{-1} \Mc{B},
\end{eqnarray}
where $\Mc{B}=\nabla_{\vec{z},\vec{\theta}} J(\vec{z}^\star,\thetabar)$ denotes the Jacobian of $\nabla_{\vec{z}} J$ with respect to $\vec{\theta}$, and $\Mc{H}=\nabla_{\vec{z},\vec{z}} J(\vec{z}^\star,\thetabar)$ denotes the Hessian of $J$ with respect to $\vec{z}$, each evaluated at $\vec{z}=\vec{z}^\star$ and $\vec{\theta}=\thetabar$. Then we may interpret $\Mc{G}'(\thetabar) \widehat{\V\theta}$ as the change in the optimal solution of \eqref{inv_prob_RS} when $\thetabar$ is perturbed in the direction $\widehat{\vec{\theta}}$, i.e., the sensitivity of the MAP point $\vec{z}^\star$ to the perturbation $\widehat{\vec{\theta}}$.

HDSA efficiently interrogates the post-optimality sensitivity operator $ \Mc{G}'(\thetabar)$ by leveraging tools such as adjoint-based derivative calculations and randomized linear algebra. We also note that post-optimality sensitivities and HDSA may be formally developed for infinite dimensional problems in a full space optimization framework. The reader is directed to \cite{HDSA} for additional details. 

\section{Interpretation of Sensitivities as Correlations in Joint Posterior} \label{sec:bayes_interp}

Although HDSA was developed in the context of PDE-constrained optimization, we establish a statistical interpretation of the sensitivities by linking them to the Bayesian inverse problem in (\ref{eqn:joint}). As shown in~\cite{gentle_intro_bayes}, the inverse Hessian of the negative log posterior PDF equals the covariance matrix of the posterior distribution when $\V{f}$ is linear. Though this does not hold for general nonlinear problems, the Laplace approximation linearizes $\V{f}$ around the MAP point and approximates the posterior covariance~\cite{isaac_2015} using the inverse Hessian evaluated at the MAP point. Similar logic is followed to establish an interpretation of the post-optimality sensitivities as a correlation in the joint posterior distribution of $(\vec{z},\vec{\theta})$. We provide an analytical result for linear inverse problems and then argue for its local validity in nonlinear inverse problems. 

Consider the Bayesian inverse problem on the joint distribution of $(\vec{z},\vec{\theta})$ and assume that the parameter-to-observable map $\V{f}(\vec{z},\vec{\theta})$ is linear in $(\vec{z},\vec{\theta})$, i.e. $\V{f}(\vec{z},\vec{\theta}) = \M{A} \vec{z} + \M{B} \vec{\theta}$ for given matrices $\M{A} \in \mathbb R^{d \times m}$ and $\M{B} \in \mathbb R^{d \times n}$. 


 
 Since $\V{f}$ is linear, and the prior is Gaussian, the joint posterior distribution $\pi_{\rm joint}(\V{z},\V\theta|\V{d})$ is Gaussian with covariance matrix $\M\Gamma_{\rm joint}$ which satisfies
 \[ \M\Gamma_{\rm joint}^{-1} =  \bmat{ \M{A}^T  \\ \M{B}^T} \M\Gamma_{\rm noise}^{-1} \bmat{\M{A} & \M{B} }  + \bmat{ \M\Gamma_{\V{z}}^{-1}  \\ &  \M\Gamma_{\V{\theta}}^{-1} }  = \bmat{ \M{A}^T  \M\Gamma_{\rm noise}^{-1}\M{A} + \M\Gamma_{\V{z}}^{-1}  & \M{A}^T \M\Gamma_{\rm noise}^{-1}\M{B}  \\ \M{B}^T \M\Gamma_{\rm noise}^{-1}\M{A} & \M{B}^T  \M\Gamma_{\rm noise}^{-1}\M{B} + \M\Gamma_{\V{\theta}}^{-1}  }  . \]
 
Theorem~\ref{thm:bayes_linear} provides an interpretation of HDSA in the context of the posterior covariance.

 \begin{theorem}
 \label{thm:bayes_linear}
 Assume that the forward model is linear; that is, $\V{f}(\vec{z},\vec{\theta}) = \M{A} \vec{z} + \M{B} \vec{\theta}$, and denote joint posterior covariance matrix as
 \[ \M\Gamma_{\rm joint} := \bmat{ \M\Gamma_{\V{z},\V{z}} &  {\M\Gamma_{\V{z},\V{\theta}}} \\ \M\Gamma_{\V{z},\V{\theta}}^T & \M\Gamma_{\V{\theta},\V{\theta}}  } . \]
 Then the post-optimality sensitivity operator $\Mc{G}'(\thetabar)$ corresponding to the negative log of the conditional posterior~\eqref{eqn:cond} satisfies
\begin{eqnarray*}
\Mc{G}'(\thetabar) = {\M{\Gamma}}_{\vec{z},\vec{\theta}} {\M\Gamma}_{\vec{\theta},\vec{\theta}}^{-1}.
\end{eqnarray*}
 \end{theorem}
 
  \begin{proof}
Define $\M{K} := \M{A}^T  \M\Gamma_{\rm noise}^{-1}\M{A} + \M\Gamma_{\V{z}}^{-1}$, $\M{L} := \M{A}^T \M\Gamma_{\rm noise}^{-1}\M{B}$, and $\M{M} := \M{B}^T  \M\Gamma_{\rm noise}^{-1}\M{B} + \M\Gamma_{\V{\theta}}^{-1}$, which represent the $(1,1)$, $(1,2)$, and $(2,2)$ blocks of $\M\Gamma_{\rm joint}^{-1}$, respectively. 
Using the block inversion formula gives
 \[ \M\Gamma_{\rm joint} = \bmat{ \M\Gamma_{\V{z},\V{z}} &  {\M\Gamma_{\V{z},\V{\theta}}} \\ \M\Gamma_{\V{z},\V{\theta}}^T & \M\Gamma_{\V{\theta},\V{\theta}}  } = \bmat{\M{K}^{-1} + \M{K}^{-1}\M{LS}^{-1}\M{L}^T \M{K}^{-1}  & - \M{K}^{-1}\M{LS}^{-1} \\ -\M{S}^{-1}\M{L}^T\M{K}^{-1} & \M{S}^{-1}   }, \]
 where $\M{S} := \M{M}^{-1} - \M{L}^T\M{K}^{-1}\M{L}$ is the Schur complement of $\M{M}$. Equating the individual subblocks, we have $\M\Gamma_{\V{z},\V{\theta}} = -\M{K}^{-1}\M{LS}^{-1} $ and $\M\Gamma_{\V{\theta},\V{\theta}} = \M{S}^{-1}$, so that $\M\Gamma_{\V{z},\V{\theta}}\M\Gamma_{\V{\theta},\V{\theta}}^{-1}=-\M{K}^{-1}\M{L}$.  

 From calculus we have $\nabla_{\V{z},\V{z}}J = \M{A}^T\M\Gamma_{\rm noise}^{-1}\M{A} + \M\Gamma_{\V{z}}^{-1} $ and $\nabla_{\V{z},\V{\theta}}J = \M{A}^T\M\Gamma_{\rm noise}^{-1}\M{B}$, which implies that
 \[ \Mc{G}'(\thetabar) = - (\M{A}^T\M\Gamma_{\rm noise}^{-1} + \M\Gamma_{\V{z}}^{-1})^{-1} \M{A}^T\M\Gamma_{\rm noise}^{-1}\M{B} = -\M{K}^{-1} \M{L}={\M{\Gamma}}_{\vec{z},\vec{\theta}} {\M\Gamma}_{\vec{\theta},\vec{\theta}}^{-1}. \]
 \end{proof}
  
Theorem~\ref{thm:bayes_linear} implies that the sensitivity of the MAP point for $\vec{z}$ with respect to $\vec{\theta}$, for linear Bayesian inverse problems with Gaussian noise and priors, corresponds to the posterior covariance between $\vec{z}$ and $\vec{\theta}$, scaled by the variance of $\vec{\theta}$. This provides a foundation to interpret HDSA in terms of the correlation between $\vec{z}$ and $\vec{\theta}$ in their joint posterior distribution. Correlations are a considerable challenge in joint inversion, and hence this interpretation of HDSA provides critical posterior information to enable better characterization of uncertainty. 

Following similar principles of the Laplace approximation (e.g., \cite{isaac_2015}), HDSA may be interpreted as local correlations in the joint posterior distribution for nonlinear inverse problems. As in \cite{isaac_2015}, we consider a linear approximation of $\V{f}$ around $(\vec{z}^\star,\thetabar)$, i.e. the linear inverse problem whose parameter-to-observable map is, to a first order approximation, 
$$\V{f}(\V{z},\V{\theta}) \approx \V{f}(\vec{z}^\star,\thetabar) + \V{f}_{\vec{z}}'(\vec{z}^\star,\thetabar)(\vec{z}-\vec{z}^\star) + \V{f}_{\vec{\theta}}'(\vec{z}^\star,\thetabar)(\vec{\theta}- \thetabar),$$ 
where $\V{f}_{\vec{z}}'$ and $\V{f}_{\vec{\theta}}'$ denote the Jacobians of $\V{f}$. Applying Theorem~\ref{thm:bayes_linear} to this linearized problem, $\Mc{G}'(\thetabar)$ measures the correlation between $\vec{z}$ and $\vec{\theta}$ in the region where the linear approximation is valid, i.e. in a neighborhood of $(\vec{z}^\star,\thetabar)$. Example~\ref{bayes_example} below further explores the connection between HDSA and the joint posterior distribution. 

\begin{example}
\label{bayes_example} To illustrate the relationship between HDSA and correlations in the joint posterior distribution, consider the synthetic model problem
$$f(z,\theta_1,\theta_2)=e^{\frac{1}{10}z\theta_1}+\theta_2.$$
We generate three data points by evaluating $f$ at the ``true" parameters $z^\dagger = 5$ and $\t^\dagger = (\theta_1, \theta_2)^T = (5, 1)^T$ and contaminate them with Gaussian noise. Taking a joint Gaussian prior on $(z,\theta_1,\theta_2)$ with mean and covariance matrix
\begin{equation*}
\vec{\mu} = 
\begin{pmatrix}
5 \\
5 \\
1 \\
\end{pmatrix}, \qquad
\mgamma = 
\begin{pmatrix}
5^2 & 0 & 0\\
0 & 2^2 & 0 \\
0 & 0 & .2^2 \\
\end{pmatrix},
\end{equation*}
respectively, we consider joint inversion on $z$ and $\t$. Figure~\ref{fig:bayes} displays the objective function $J$ with $\theta_2$ fixed at $1$ (left panel) and $\theta_1$ fixed at $5$ (right panel) with MAP points of the conditional posteriors (for different $\vec{\theta}$'s) and hyper-differential sensitivities indicated by dots. 
\begin{figure}[h]
\centering
  \includegraphics[width=0.49\textwidth]{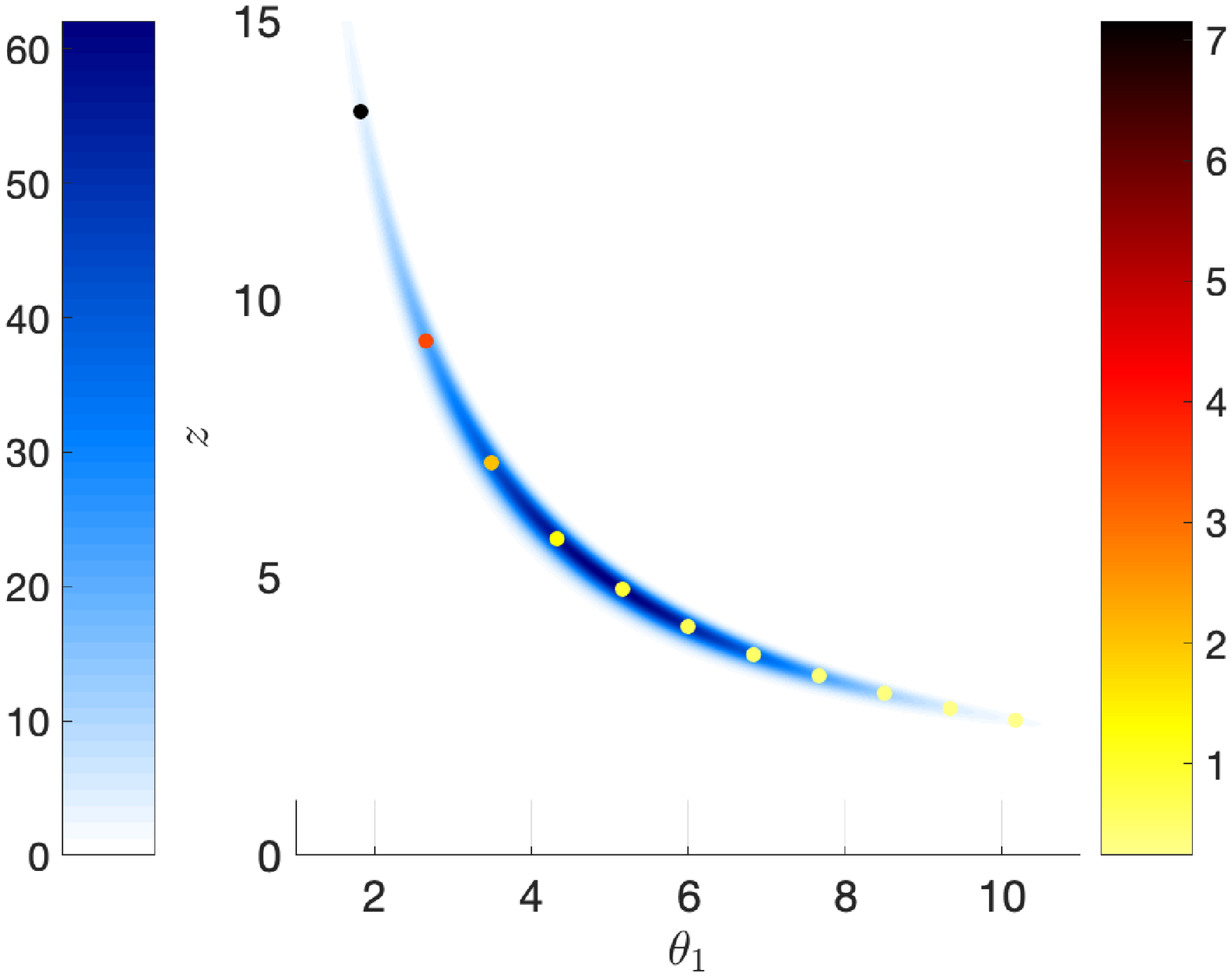}
    \includegraphics[width=0.49\textwidth]{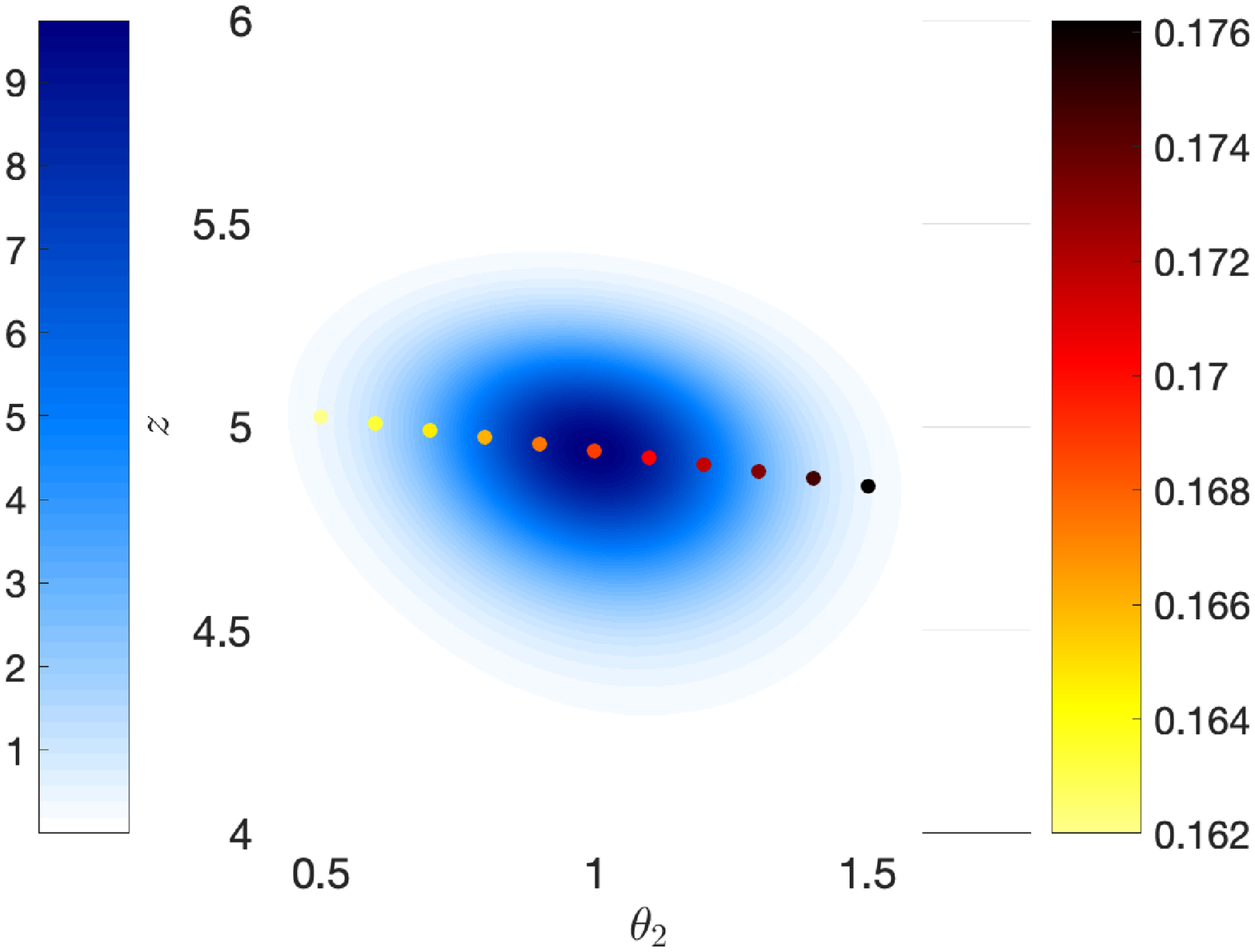}
  \caption{Objective function $J$ with $\theta_2$ fixed at $1$ (left panel) and $\theta_1$ fixed at $5$ (right panel). In each panel, the shading (in the blue color scale) indicates the joint objective while the solid dots (in the red-yellow color scale) denote the maximum a posterior probability (MAP) point $z^\star$ of $\pi_{\rm cond}(z|d;\bar{\t})$. The color of the dot (in the red-yellow color scale) indicates the hyper-differential sensitivities for the MAP point with respect to $\theta_1$ (left panel) and $\theta_2$ (right panel).}
    \label{fig:bayes}
\end{figure}

We observe several trends which illustrate the interpretation of HDSA:
\begin{enumerate}
\item[$\bullet$] The joint objective indicates a stronger dependency between $z$ and $\theta_1$ in comparison to $z$ and $\theta_2$. Corresponding to this, we observe a much larger sensitivity  for $\theta_1$ (given by the red-yellow color bar in the left panel) than $\theta_2$ (given by the red-yellow color bar in the right panel). 
\item[$\bullet$] The sensitivities vary with $\vec{\theta}$ showing that it is measuring local correlations which vary with $\vec{\theta}$. We observe the greatest sensitivity for small values of $\theta_1$ which is consistent with the visual observation that the MAP point of $\pi_{\rm cond}(z|\V{d};\bar{\t})$ varies more for small values of $\theta_1$.
\item[$\bullet$] There is little variation in the sensitivities in the right panel where the distribution is approximately Gaussian (since local correlations do not vary in Gaussian distributions) and much greater variation in the left panel where the distribution is highly non-Gaussian. 
\end{enumerate}
\end{example}


\section{HDSA and Likelihood Informed Subspaces}\label{sec:lis}
Having motivated the post-optimality sensitivities for Bayesian inverse problems, this section recalls the computation of sensitivities in the likelihood informed subspace~\cite{hdsa_ill_posed_inv_prob} and demonstrates how this computation may be done concurrently with approximate posterior sampling via the Laplace approximation. 

\subsection{Computing the projected sensitivity indices} 
For ill-posed inverse problems, the post-optimality sensitivities may be dominated by the directions which are poorly informed by data. These directions are given by the eigenvectors corresponding to the smallest eigenvalues of $\Hop$. To compute sensitivities in the directions informed by data, we follow~\cite{hdsa_ill_posed_inv_prob} and introduce the projection of the sensitivity operator onto the likelihood informed subspace (LIS). 


The LIS was used for dimension reduction in nonlinear Bayesian inverse problems in \cite{cui_2014}. It defines the informed directions by the $r$ leading eigenvectors of the generalized eigenvalue problem 
\begin{equation}\label{eq:gevp}
\Hop_{M} \v_j = \lambda_j \M\Gamma_{\V{z}}^{-1} \v_j, \quad j = 1,2,\dots,m.
\end{equation}
where 
\begin{equation}\label{eq:misfit_hessian}
\Hop_{M} = \nabla_{\z\z} M(\z^\star,\thetabar)
\end{equation}
with $M$ as defined in~\eqref{eqn:M_and_R}. Multiplying~\eqref{eq:gevp} by $\v_j^T$ to compute the Rayleigh quotient 
\begin{equation*}
\lambda_j  = \frac{\v_j^T \Hop_{M} \v_j }{\v_j^T \M\Gamma_{\V{z}}^{-1} \v_j}
\end{equation*}
provides an interpretation of the eigenvalue as a ratio of contributions from the likelihood (since $\Hop_{M}$ is the negative log likelihood Hessian) and the prior. The eigenvectors define the corresponding directions, i.e. the eigenvectors corresponding to large eigenvalues specify directions in the parameter space which are informed by the likelihood.

The truncation rank $r$ is application dependent and dictates the computational complexity. For ill-posed problems, $r$ may be small as a result of data sparsity. In this case, one can use efficient Krylov or randomized methods to solve the eigenvalue problem~\cite{hdsa_ill_posed_inv_prob}. To solve~\eqref{eq:gevp}, we employ the ``Two pass" randomized algorithm introduced in~\cite{arvind}. We define a projector onto the LIS as 
\begin{equation}\label{eq:lis_projector}
\Pop = \Vop_r\Vop_r^T\M\Gamma_{\V{z}}^{-1},
\end{equation}
where the columns of $\Vop_r$ are given by the $r$ leading eigenvectors in~\eqref{eq:gevp}. Then the sensitivities in directions informed by the likelihood, see \cite{hdsa_ill_posed_inv_prob} for details, are given by
\begin{equation}\label{eq:lis_sens_indices}
S_i = \norm{\Pop \Hop^{-1}\Bop\e_i}_{\M{W}_Z} = \sqrt{\sum_{k=1}^r\sum_{j=1}^r\left(\frac{\v_j^T\Bop\e_i}{1 + \lambda_j}\right)\left(\frac{\v_k^T\Bop\e_i}{1 + \lambda_k}\right)\v_k^T\M{W}_Z\v_j},\qquad  i=1,2,\dots, n,
\end{equation}
where $\e_i$ denotes the $i^{th}$ canonical basis vector and $\M{W}_Z$ is a symmetric positive definite matrix defining a weighted inner product, for instance, a mass matrix if $z$ is discretized with finite elements. We interpret $S_i$ as the influence of the $i^{th}$ parameter on the MAP point projected onto the LIS. Equation~\eqref{eq:lis_sens_indices} follows from considering a spectral representation of the $\Hop$ in combination with the Sherman-Morrison-Woodbury formula. 


\subsection{Gaussian Approximation of the Posterior}\label{sec:gauss_approx}
If the mapping from the parameter space to observations $\V{f}$ is nonlinear, the posterior distribution is not Gaussian. As a result, the posterior cannot conveniently be described in terms of its mean and covariance. This motivates alternative methods for describing the posterior. State-of-the-art  Markov chain Monte Carlo (MCMC) methods may be used to exactly specify the posterior via sampling. However, MCMC methods often require repeated evaluations of the parameter-to-observable map. In the applications that we are interested in, $\V{f}$ maps onto observations of a (possibly) nonlinear PDE. For computational efficiency, we use the Laplace approximation to the posterior distribution (see e.g.,~\cite{isaac_2015}). Specifically, we assume the posterior is approximately Gaussian with mean $\z^\star$ and covariance matrix $\mgamma_\text{post} = \Hop^{-1}$, where $\Hop$ is the Hessian of the objective function \eqref{inv_prob_RS} evaluated at $\z^\star$. We write the approximate posterior as $\hat{\pi}_{\rm cond}$ and define it as
\begin{equation}\label{eq:approx_post}
\hat{\pi}_{\rm cond}(\V{z} \vert \V{d}; \thetabar) \sim \mathcal{N}(\z^\star,\mgamma_\text{post}).
\end{equation}

Explicit construction of the Hessian/covariance matrix is infeasible as they are only accessible via matrix-vector products. Specifically, computing $\Hop \V{x}$, for $\V{x} \in \r^m$, requires two linear PDE solves (called the incremental state and incremental adjoint equations), and computing  $\Hop^{-1} \V{x}$, for $\V{x} \in \r^m$, requires an iterative solver needing multiple matrix-vector products with $\Hop$.  Thus we resort to a low-rank approximation of the Hessian to efficiently store and compute with the posterior covariance matrix.

\subsubsection{Efficient representation of posterior covariance}\label{sssec:low_rank_approx}
The Hessian $\Hop$ is the sum of the data misfit Hessian $\Hop_M$ evaluated at $\z^\star$ and $\mgamma^{-1}_{\V{z}}$.  
In many ill-posed inverse problems of interest, the data misfit Hessian when appropriately weighted by the prior covariance matrix exhibits fast spectral decay. In particular, the fast decay is a symptom of directions in the parameter space poorly informed by the data~\cite{flath2011fast} and as a result, we can represent the data misfit Hessian with a low rank approximation. Following the procedure of~\cite{Villa_2021}, the low-rank approximation is obtained by first solving the LIS generalized eigenvalue problem~\eqref{eq:gevp} and retaining the first $r$ eigenvalues so that $1 \gg \lambda_r $. After applying the Sherman-Morrison-Woodbury formula we can write
\begin{equation}\label{eq:hess_inv}
\Hop^{-1} = \left(\mgamma^{-1}_{\V{z}} + \Hop_{M}\right)^{-1} =  \mgamma_{\V{z}}- \Vop_r\md_r\Vop_r^T + \mathcal{O}\left(\sum_{i=r+1}^m \frac{\lambda_i}{1 + \lambda_i}\right),
\end{equation}
where $\md_r = \text{diag}(\lambda_1/(\lambda_1+1),\dots,\lambda_r/(\lambda_r+1))$. We define the approximate posterior covariance as
\begin{equation}\label{eq:approx_cov} 
\mgamma_\text{post} = \mgamma_{\V{z}} - \Vop_r\md_r\Vop_r^T.
\end{equation}
Equation~\eqref{eq:approx_cov} shows how the eigenpairs calculated for the sensitivity indices can be used for the posterior approximation. This implies that approximate posterior sampling is a computational by-product of computing sensitivities. In the next subsection, we describe the methods for quantifying the uncertainty in the solution through sampling and variance.

\subsubsection{Sampling from approximate posterior distribution and posterior variance}\label{sssec:sampling_variance}
Recall that the posterior distribution represents the information gained by combing the data (likelihood) and our beliefs (prior). To represent the reduction of uncertainty we draw samples from the prior and posterior and compute their variances (the diagonals of $\mgamma_{\V{z}}$ and $\mgamma_\text{post}$).
 
A sample $\z$ from $\mathcal{N}(\V{\mu}_{\V{z}}, \mgamma_{\V{z}})$ can be computed via the formula 
\begin{equation}\label{eq:gen_sampling}
\z_\text{prior} = \V{\mu}_{\V{z}} + \mc\x 
\end{equation}
where $\mc$ is any matrix satisfying $\mgamma_{\V{z}} = \mc\mc^T $ and $\x \sim \pN(\zero,\mi)$. The matrix $\mc$ can be calculated in a multitude of ways. One method is to compute $\mc$ as a Cholesky factor.  In the case of large-scale problems computing the Cholesky factor is not feasible. To remedy this, we follow the preconditioned Lanczos method for sampling multivariate Gaussian distributions proposed by Chow and Saad in \cite{Chow_2014}. This method utilizes the Lanczos algorithm to approximate the action of $\mgamma^{1/2}_{\V{z}}$ onto a vector. 

To sample from the Laplace approximation of the posterior $\pN(\z^\star,\mgamma_\text{post})$ we follow the approach described in Section 4.3.2 of  \cite{Villa_2021}. Given a sample $\z_\text{prior}$ from the prior distribution, a sample from  $\pN(\z^\star,\mgamma_\text{post})$ is given by 
\begin{equation}\label{eq:gen_sampling_post}
\z_\text{post} = \left(\mi_m - \Vop_r\ms_r\Vop_r^T\mgamma^{-1}_{\V{z}}\right)(\z_\text{prior}-\V{\mu}_z) + \z^\star,
\end{equation}
where $\ms_r = \mi_r - (\mlambda_r + \mi_r)^{-\frac12}$.

Since our approximate posterior distribution is Gaussian with covariance $\mgamma_\text{post} =  \M\Gamma_{\V{z}} - \Vop_r\md_r\Vop_r^T$, the posterior variance is given by 
\begin{equation}\label{eqn:postvar}
 \text{diag}(\mgamma_\text{post}) =  \text{diag}(\mgamma_{\V{z}}) - \text{diag}(\Vop_r\md_r\Vop_r^T).
\end{equation}
To approximate $\text{diag}(\mgamma_{\V{z}})$, 
we use the Diag++ algorithm~\cite[Algorithm 1]{Baston2022StochasticDE}. Given a budget of $s_D$ matrix-vector products, Diag++ approximates the diagonal of $\mgamma_{\V{z}}$ by sketching its range space using $s_D/3$ matrix-vector products, computing the diagonal of a low rank approximation using $s_D/3$ matrix-vector products, and estimating the diagonal of $\mgamma_{\V{z}}$ minus its low rank approximation using $s_D/3$ matrix-vector products. Since $\mgamma_{\V{z}}$ is the inverse of a differential operator, the main cost of Diag++ is $s_D$ linear solves. 

The diagonal of $\Vop_r\md_r\Vop_r^T$ is calculated directly via the formula 
\begin{equation}\label{eq:diag_exact}
\text{diag}(\Vop_r\md_r\Vop_r^T) = \sum_{k=1}^r \left(\frac{\lambda_k}{1 + \lambda_k} \v_k \right) \odot \v_k.
\end{equation}

\section{Algorithms}\label{sec:algorithms}

In this section we show our algorithm for computing HDSA indices and the Gaussian approximation. The algorithm consists of four major components: the MAP point, the generalized eigenvalue problem, the sensitivity indices, and the Gaussian approximation. The MAP point is key to performing HDSA and computing the Gaussian approximation of the posterior (as seen in Section~\ref{sec:lis}). In the subsections below, we summarize the computational details of each major component. Note that this discussion is based on the assumptions for the data $\d$, prior $\z$, and auxiliary parameters $\t$ as stated in Section~\ref{ssec:bayes_inv_probs}. 

\subsection{Computation of the MAP point}
Recall that the MAP point is the solution of the optimization problem~\eqref{inv_prob_RS}. This is a nonlinear least squares problem and we employ the truncated Newton CG trust region method~\cite{Kouri2018}. This corresponds to the loop in Lines~\ref{alg:tr_loop_start}-\ref{alg:tr_loop_end} in Algorithm~\ref{alg:full}. Given an iterate $\z_k$, we compute the gradient of the objective function at $\z_k$, $\g_k$ (which requires a forward and adjoint PDE solve) and check for convergence to a stationary point. If $\z_k$ is not a stationary point, we form the local model of the objective function $\m_k(\s) = J(\s) + \g_k^T\s +  \frac{1}{2}\s^T\mh_k\s$ where $\mh_k$ is the Hessian of $J$ evaluated at $\z_k$. Then the trust region subproblem is solved using the CG-Steihaug algorithm. Note that we compute exact Hessian vector products using incremental adjoint equations. To calculate $\z_{k+1}$ the reduction in the local model is considered and then the trust region is updated. As stated in the introduction, the MAP is only a point estimate of the unknown and is not sufficient for describing relationships between the unknown and auxiliary parameters. To address this relationship we employ HDSA to the optimization problem~\eqref{inv_prob_RS}.

\begin{algorithm}[!ht]
\DontPrintSemicolon
\SetAlgoNoLine
\caption{Workflow for HDSA in Bayesian inverse problems}\label{alg:full}
\KwIn{$r_0 \in \mathbb{N}$, $\Delta r \in \mathbb{N}$, $p\in\mathbb{N}$, and $\lambda_\text{min} \in \r$.} 
\hrulefill
\\
\textbf{Solve MAP point} \qquad $\min\limits_{\vec{z} \in \mathbb R^m} J(\vec{z}, \thetabar):=M(\vec{z}, \thetabar ) + R(\vec{z})$ \;
\textit{Truncated Newton CG trust region procedure}: \;
\For{$k=0,1,2,\dots$} {  \nllabel{alg:tr_loop_start} 
Compute $\g_k = \nabla_{\z} J$ and $\M{H}_k=\nabla_{\z,\z} J$ to form a local model $\m_k$  \nllabel{alg:pde_solves} \;
Minimize the local model $\m_k$ subject to the trust region constraint \;
Update trust region radius \;
} \nllabel{alg:tr_loop_end}
\hrulefill 
\\
\textbf{Compute LIS eigenpairs} \qquad $\Hop_M \v_j = \lambda_j \mgamma^{-1}_{\V{z}} \v_j, j =1,2,\dots,r$
\\
\textit{Randomized Generalized Hermitian Eigenvalue Procedure}: \;
Set $\lambda_\text{iter} = \infty$, Sample a random matrix $\momega \in \r^{n\times (r_0-\Delta r + p)}$ \nllabel{alg:start_eig} \;
\While{$\lambda_\text{iter} > \lambda_\text{min}$} {
Augment $\momega$ with $\Delta r$ additional columns \;
Compute: \;
\quad sketch $\my = \mgamma_{\V{z}}\Hop_M\momega$\nllabel{alg:sketch} \;
\quad projection $\mt$ onto range of $\my$\nllabel{alg:projection}\;
\quad eigendecomposition $\mt = \ms\mlambda\ms^T$\nllabel{alg:eigen_decomp} \;
\quad Set $\lambda_\text{iter} = (\mlambda)_{r,r}$\nllabel{alg:smallest} \;
}
Compute $\v_j = \mq\s_j$, $1 \leq j \leq r$, where $\s_j$ is the $jth$ column of $\ms$\nllabel{alg:end_eig} \;
\hrulefill
\\
\textbf{Calculate Sensitivity indices} \qquad $S_i = \norm{\Pop \Hop^{-1}\Bop\e_i}_{\M{W}_Z}$ for $1 \leq i \leq n$ \;
Compute: \;
\quad $\Bop\e_i = \nabla_{\z,\t} J_\text{\rm post}(\z^\star, \thetabar)\e_i$, $i=1,2,\dots,n$\label{alg:b_ei} \;
\quad $S_i$, $i=1,2,\dots,n$, using~\eqref{eq:lis_sens_indices}  \label{alg:sens} \;
\hrulefill
\\
\textbf{Laplace Approximation} \qquad $\hat{\pi}_\text{post}(\z | \d; \thetabar) \sim \mathcal{N}(\z^\star,\mgamma_\text{post}), \quad \mgamma_\text{post} = \mgamma_{\V{z}} - \Vop_r\md_r\Vop_r^T $ \label{alg:laplace} 
\\
\textit{Sampling Procedure}: \qquad $\z_\text{post} = \left(\mi_m - \Vop_r\ms_r\Vop_r^T\mgamma^{-1}_{\V{z}}\right)(\z_\text{prior}-\V{\mu}_z)+ \z^\star$, where $\z_\text{\rm prior}$ is generated from~\eqref{eq:gen_sampling}
\\
\textit{Variance Procedure}: Estimate $\text{diag}(\mgamma_\text{post})$ using~\cite[Algorithm 1]{Baston2022StochasticDE} and~\eqref{eq:diag_exact} \label{alg:var} 
\end{algorithm}

\subsection{Computing the eigenpairs}\label{ssec:hgevp}
Solving the Hermitian generalized eigenvalue problem (HGEVP)~\eqref{eq:gevp} is important to computing the LIS for HDSA and the Gaussian approximation of the posterior. In this subsection, we elaborate the calculations of the HGEVP that correspond to Lines~\ref{alg:start_eig}-\ref{alg:end_eig} in Algorithm~\ref{alg:full}. 

The Randomized Generalized Hermitian Eigenvalue Procedure is a modification of the ``Double Pass'' algorithm~\cite{arvind}. Line~\ref{alg:sketch} sketches the range of $\mgamma_{\V{z}}\Hop_M$ by multiplying it by a collection of random vectors $\momega$. Next, a $\mgamma_{\V{z}}^{-1}$ orthogonal basis for the range of $\my$ is computed using dense linear algebra. The second round of matrix-vector products with $\Hop_M$ compute the low rank projection onto the range space of $\my$.  Dense linear algebra is done to compute the eigendecomposition of the projection $\mt$. If the smallest eigenvalue is above the target eigenvalue threshold then the loop continues by augmenting the columns of $\momega$ or else the loop ends. The desired eigenvectors are obtained through $\v_j = \mq\s_j$. The sketch dominates the computational cost of the procedure and requires $2(r + p)$ matrix-vector product with $\Hop_M$ where $r$ is the final target rank and $p$ is the oversampling parameter. Note that several of the computations in the Randomized Generalized Hermitian Eigenvalue Procedure can be easily parallelized. The details of this procedure are given in \cite{hdsa_ill_posed_inv_prob}. 

\subsection{Computing the Sensitivity Indices}\label{ssec:sens_ind}
The sensitivity indices calculation corresponds to lines~\ref{alg:b_ei}-\ref{alg:sens}. Since we have the eigenpairs of~\eqref{eq:gevp}, the dominant cost of the sensitivity calculation is the action of $\Bop$ onto basis vectors $\e_i$. The matrix $\Bop$ corresponds to the Jacobian of $\nabla_{\z} J(\z^\star, \thetabar)$ with respect to $\t$, so each $\Bop\e_i$ is two additional PDE solves. Therefore computing $\{\Bop\e_i\}_{i=1}^n$ costs $2n$ PDE solves where $n$ is the auxiliary parameter dimension. 

\subsection{Computing the Laplace approximation}\label{ssec:ga_comp}
Approximating the posterior utilizes the calculation of the eigenpairs in~\eqref{eq:gevp}. To sample from the posterior we draw from the prior using the Lanczos method as explained in Section~\ref{sssec:low_rank_approx}. Using this method without a preconditioner, the cost of approximating prior samples $\mgamma^{1/2}_{\V{z}} \vec{\epsilon}$, where $\vec{\epsilon}$ is a standard normal Gaussian sample, is $N_\text{lanczos}$ matrix-vector products with $\mgamma_{\V{z}}$. For the variance, the dominant cost  is the Diag++ algorithm requiring $s_D$ matrix-vector products with $\mgamma_{\V{z}}$. 


\section{Application to Ice Sheet Bedrock Inversion} \label{sec:num_results}
In this section of the paper, we apply the numerical methods and algorithms described in Sections~\ref{sec:lis}-\ref{sec:algorithms} to an application in ice sheet inversion. The goal is to use surface velocity measurements on the Greenland ice sheet to compute an estimate of the uncertain bedrock topography beneath the ice, quantify the influence of auxiliary parameters on the solution, and provide a statistical characterization. To accomplish this we first introduce the shallow ice sheet model in Section~\ref{ssec:si_model} and the inverse problem that we are solving in Section~\ref{ssec:inv_prob_formulation}. The inverse problem formulation is followed by the MAP estimate results in Section~\ref{ssec:inv_results}. The MAP estimate drives the HDSA calculation which we apply to the inversion of the bedrock with respect to the log basal friction (a spatial coefficient representing the interaction between the land ice and the bedrock) and forcing. This is detailed in Section~\ref{ssec:hdsa_results}. In Section~\ref{ssec:samp_var_results} we provide results for posterior samples and variance. 

\subsection{Ice Sheet Model}\label{ssec:si_model}
Developing high-fidelity models for ice sheets is important for global climate modeling and prediction of sea level rise. Simulation of ice sheets such as Antarctica and Greenland requires extensive computational resources to solve nonlinear partial differential equations (PDEs) on fine spatial meshes. A high-fidelity model for ice sheet dynamics is typically derived by considering the ice sheet to behave as a viscous shear-thinning fluid in a low Reynolds-number flow. This results in the nonlinear Stokes equation in three spatial dimensions coupled with equations for the temperature distribution in the ice and the thickness of the ice. From these fundamental equations, a variety of assumptions are made to yield models of different physics fidelities, a comprehensive overview of which is beyond the scope of this article. To demonstrate HDSA in this article we adopt the shallow ice approximation (SIA) \cite{morland_1980,hutter_1983} with the isothermal assumption \cite{bueler_2005}.

\subsection{Mathematical formulation of Inverse Problem}\label{ssec:inv_prob_formulation}
We consider the following inverse problem for the discretized bedrock topography $\V{b}$ of the ice sheet
\begin{eqnarray}\label{eqn:basal_inversion}
\min_{\b \in\real^m} \quad  J(\b, \thetabar) = \frac{1}{2} \norm{\V{f}(\b,\thetabar)-\vec{d}}^2_{\mgamma^{-1}_\text{noise}} + \frac12\norm{\b-\b_0}^2_{\mgamma^{-1}_{\V{b}}}  
\end{eqnarray} 
where $\b_0$, $\mgamma_{\V{b}}$, $\thetabar$, and $\d$ correspond to the prior bedrock topography mean and covariance, the nominal value of the auxiliary parameters, and the data, respectively. The forward operator $\V{f}(\b,\thetabar)$ is the mapping from the bedrock topography and the auxiliary parameters (log basal friction and forcing) to surface velocity in~\eqref{eqn:vel}.


\textbf{Shallow ice approximation:} The inverse problem \eqref{eqn:basal_inversion} is constrained by the SIA equation 
\begin{align}
\label{eqn:sia_pde_hdsa}
\frac{\partial s}{\partial t} - \nabla \cdot (Q(s)\nabla s) =& h_{\text{flux}}(\t) \qquad & \text{on } \Omega \times (0,T] \\
\nabla s \cdot n = & 0 & \text{on } \partial \Omega \times (0,T] \nonumber \\
 s =&  s_0 & \text{on } \Omega \times \{0\} \nonumber \\
Q(s) =& e^{-\gamma(\t)} \rho g (s-b)^2 + \frac{2A\rho^3 g^3}{5}(s-b)^5 \vert \vert \nabla s \vert \vert^2,
\end{align}
where $s$ is the surface height and $Q(s)$ is a velocity field derived from the Stokes velocity model. The scalar parameters $\rho = 910$ $(kg/m^3)$ is the density of ice, $g=9.81$ $(m/s^2)$ is the acceleration of gravity, $A=10^{-16}$ $(Pa^{-3}/s)$ is a flow rate factor. 

\textbf{Parameter to observable map}:
The bedrock topography $b$, log basal friction $\gamma$, and forcing $h_\text{flux}$ are mapped to the shallow ice approximation~\eqref{eqn:sia_pde_hdsa}. The resulting surface height $s$ and the bedrock topography are then mapped to the following equation for the surface velocity on the ice sheet
\begin{equation}
    \vec{v}(b,\gamma,h_\text{flux})=-\frac{1}{2} A \rho^3 g^3 (s(b,\gamma,h_\text{flux})-b)^4 \vert \vert \nabla s(b,\gamma,h_\text{flux}) \vert \vert^2 \nabla s(b,\gamma,h_\text{flux}) \label{eqn:vel}.
\end{equation}

\textbf{Domain:} We consider the inversion of bedrock topography $\b$ in a $550 \times 450$ km region of Greenland specified in Figure~\ref{fig:greenland}. The surface velocity $\v$ is calculated as a function of the surface height (the PDE~\eqref{eqn:sia_pde_hdsa} solution) via~\eqref{eqn:vel} and compared with velocity observations from satellites. 
\begin{figure}[h]
\centering
 \includegraphics[width=0.69\textwidth]{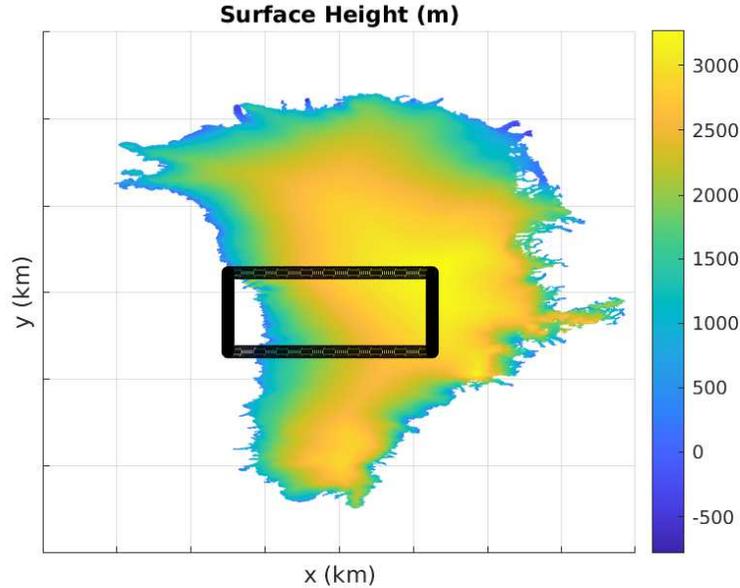}
  \caption{Ice surface height data for the Greenland ice sheet. The black box indicates the region analyzed in this article.}
  \label{fig:greenland}
\end{figure}


\textbf{Auxiliary parameters:} The log basal friction $\gamma$ and forcing $h_\text{flux}$ are given by 
\begin{equation}\label{eq:uncertain_params}
\gamma(\t) = \tilde{\gamma}\delta_{1}(\t) \qquad \text{and} \qquad \forcing(\t) = \tilde{h}_{\text{flux}}\delta_{2}(\t)
\end{equation}
where 
\begin{equation}\label{eq:perturbation}
\delta_i(\t) = 1 + 0.2\sum_{j=1}^n \theta_{((i-1)n+j)}\phi_j, \qquad i = 1,2,
\end{equation}
are parameterized perturbations; $\{\phi_j\}_{j=1}^n$ are linear finite element basis functions defined on a $31 \times 31$ mesh of the domain $\Omega$. This implies that $n = 961$ so there are $1922$ auxiliary parameters in total. The nominal parameters are $\thetabar = \zero \in \real^{1922}$ yielding $\gamma(\zero)=\tilde{\gamma}$ and $\forcing(\zero)=\tilde{h}$ (magnitudes are shown in Figure~\ref{fig:forward_model}). For simplicity, we provide a summary of relevant variables and constants for the inverse problem in Table~\ref{tab:constants}.

\textbf{Prior parameters:} We assume a Gaussian prior with mean zero and covariance $\mgamma_{\V{b}}$ given by the matrix representation of the operator $\left(-\beta \Delta + \alpha \mathcal I \right)^{-2}$ with constants $\beta =10^{-2}$ and $\alpha =9\times10^{-7}$. The Laplacian operator $\Delta$ and identity operator $\mathcal I$ are equipped with zero Neumann boundary conditions. The parameter $\beta$ controls the smoothness of prior samples and $\alpha$ controls the variance of the prior samples.

\begin{table}
\begin{center}
\begin{tabular}[h]{ c|c } 
\hline 
\hline
Constants & Values \\
 \hline
 \hline
 $\beta$ & $10^{-2}$ \\
 \hline 
$\alpha$ & $9 \times 10^{-7}$ \\
\hline 
$\mgamma_{\V{b}}$ & matrix representation of $(-\gamma\Delta + \alpha)^{-2}$ \\
\hline
$\mgamma_\text{noise}$ & $(50)^2\mi$ \\
\hline
$\d$ & surface velocity data $(m/s)$\\
\hline 
$A$ & flow rate factor $10^{-16}$ $(Pa^{-3}/s)$ \\
\hline 
$\rho$ & density of ice $910$ $(kg/m^3)$ \\
\hline 
$g$ & acceleration of gravity $9.81$ $(m/s^2)$ \\
\hline 
$\gamma$ & log basal friction \\
\hline 
$h_{\text{flux}}$ & forcing \\
\hline 
\hline
Variables & Values \\
\hline 
\hline
$b$ & bedrock topography $(m)$ \\
\hline 
$\v$ & surface velocity $(m/s)$ \\
\hline 
$s$ & surface height $(m)$ \\
\hline
\end{tabular}
\end{center}
\caption{Table of relevant variables and constants for the inverse problem~\eqref{eqn:basal_inversion}-\eqref{eq:perturbation}.}
\label{tab:constants}
\end{table}

\textbf{Synthetic data generation:} To facilitate a numerical demonstration, data is generated from the model using ``true'' bedrock topography, which is subsequently considered unknown as we solve and analyze the inverse problem to reconstruct it. The bedrock, ice thickness, accumulation/ablation, and log basal friction are taken from \cite{albany_felix}.
To aid in numerical performance we pre-smooth the log basal friction and bedrock topography using a local averaging technique. The accumulation/ablation forcing term from \cite{albany_felix} is stationary. The forcing term is depicted in the top left panel of Figure~\ref{fig:forward_model}.

\begin{figure}[h]
\centering
  \includegraphics[width=0.49\textwidth]{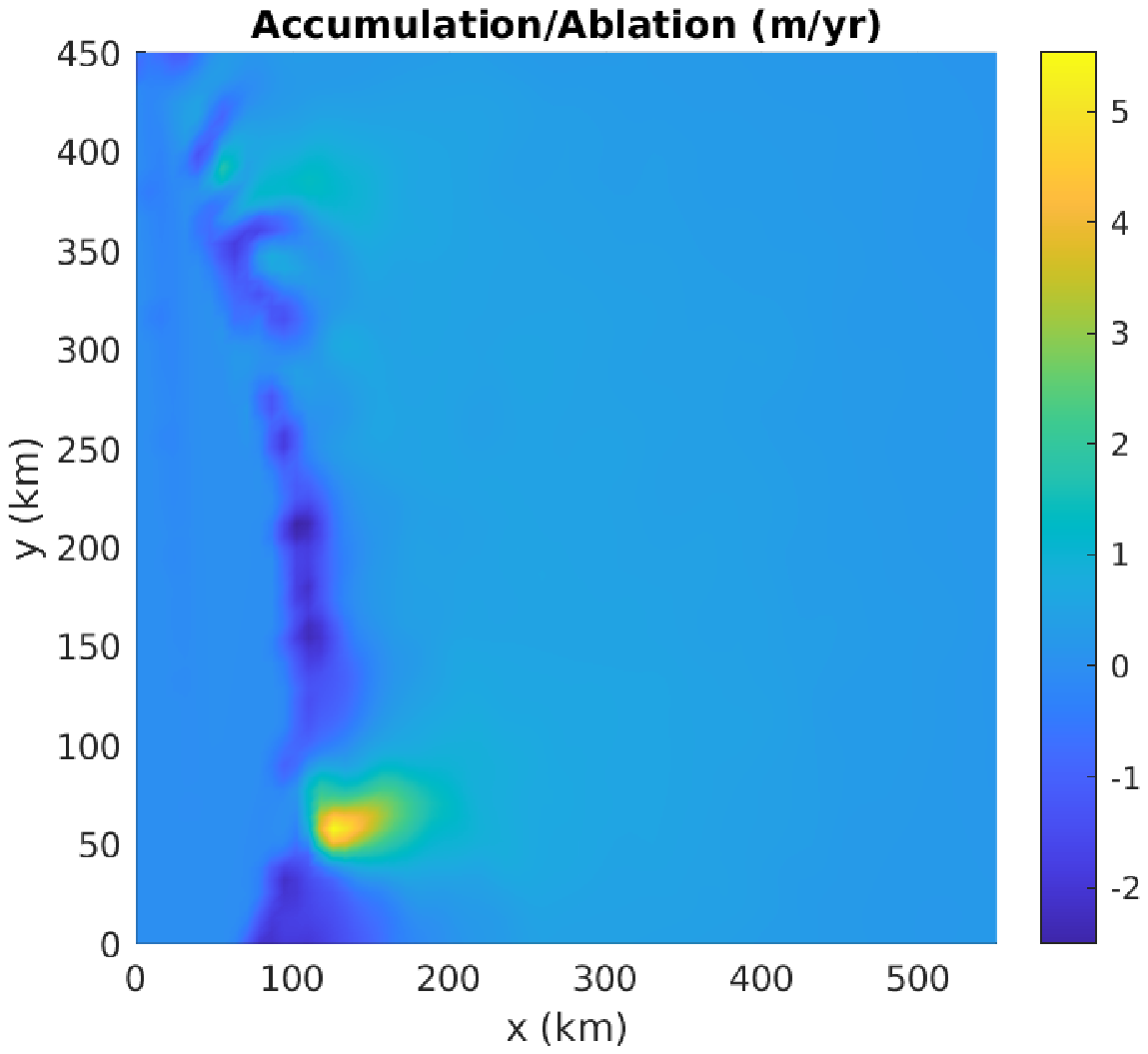}
  \includegraphics[width=0.49\textwidth]{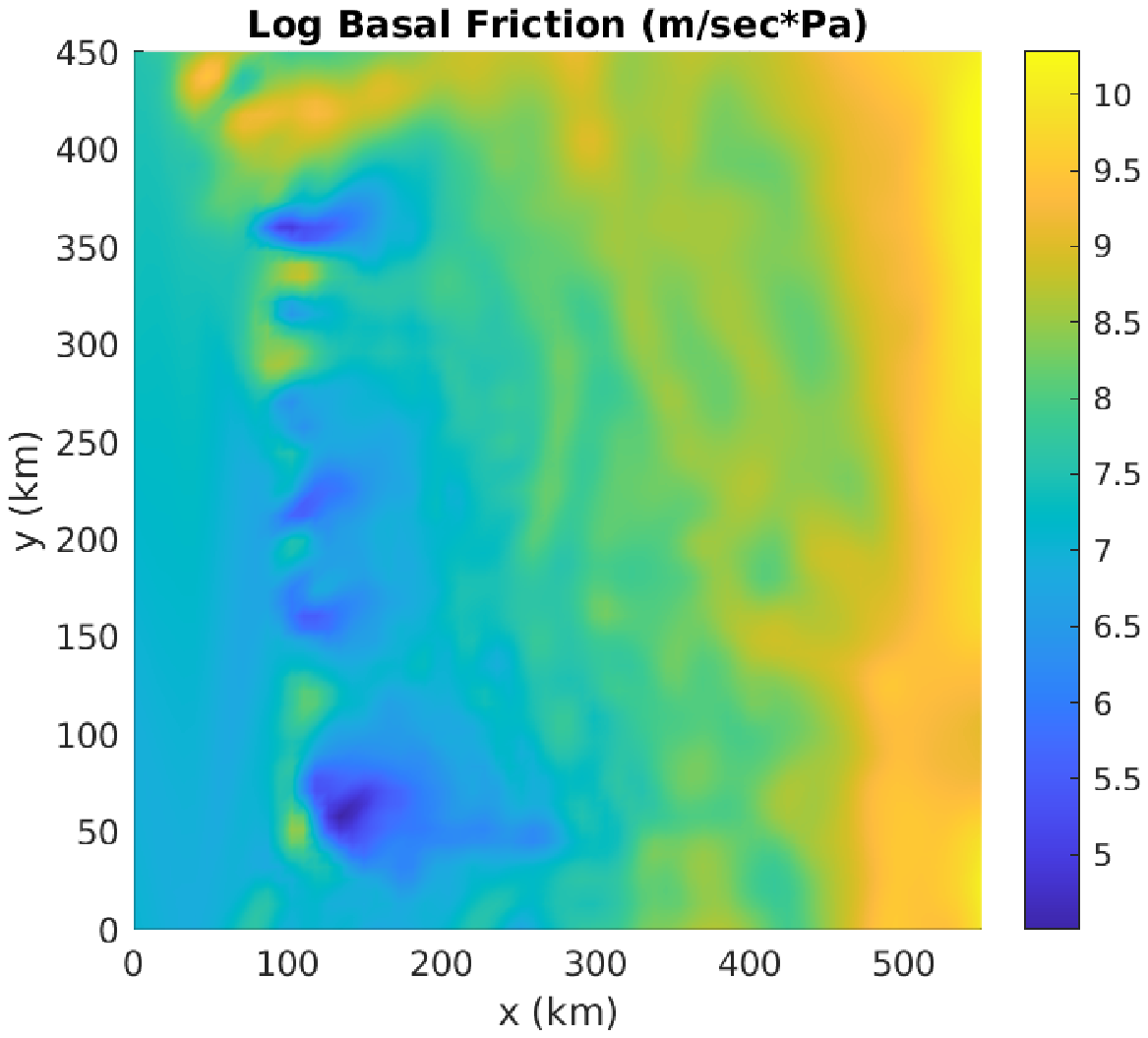}
   \includegraphics[width=0.49\textwidth]{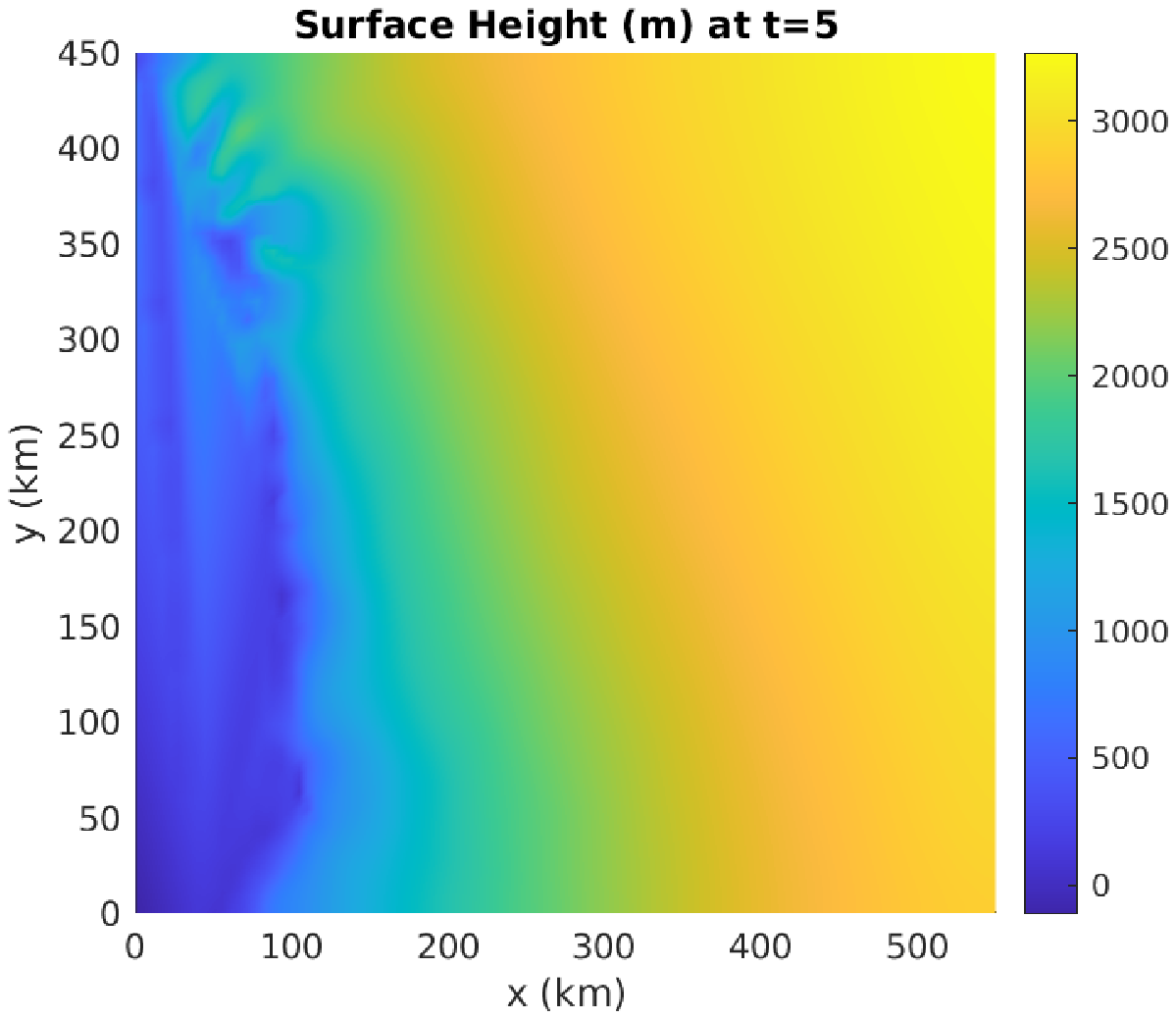}
 \includegraphics[width=0.49\textwidth]{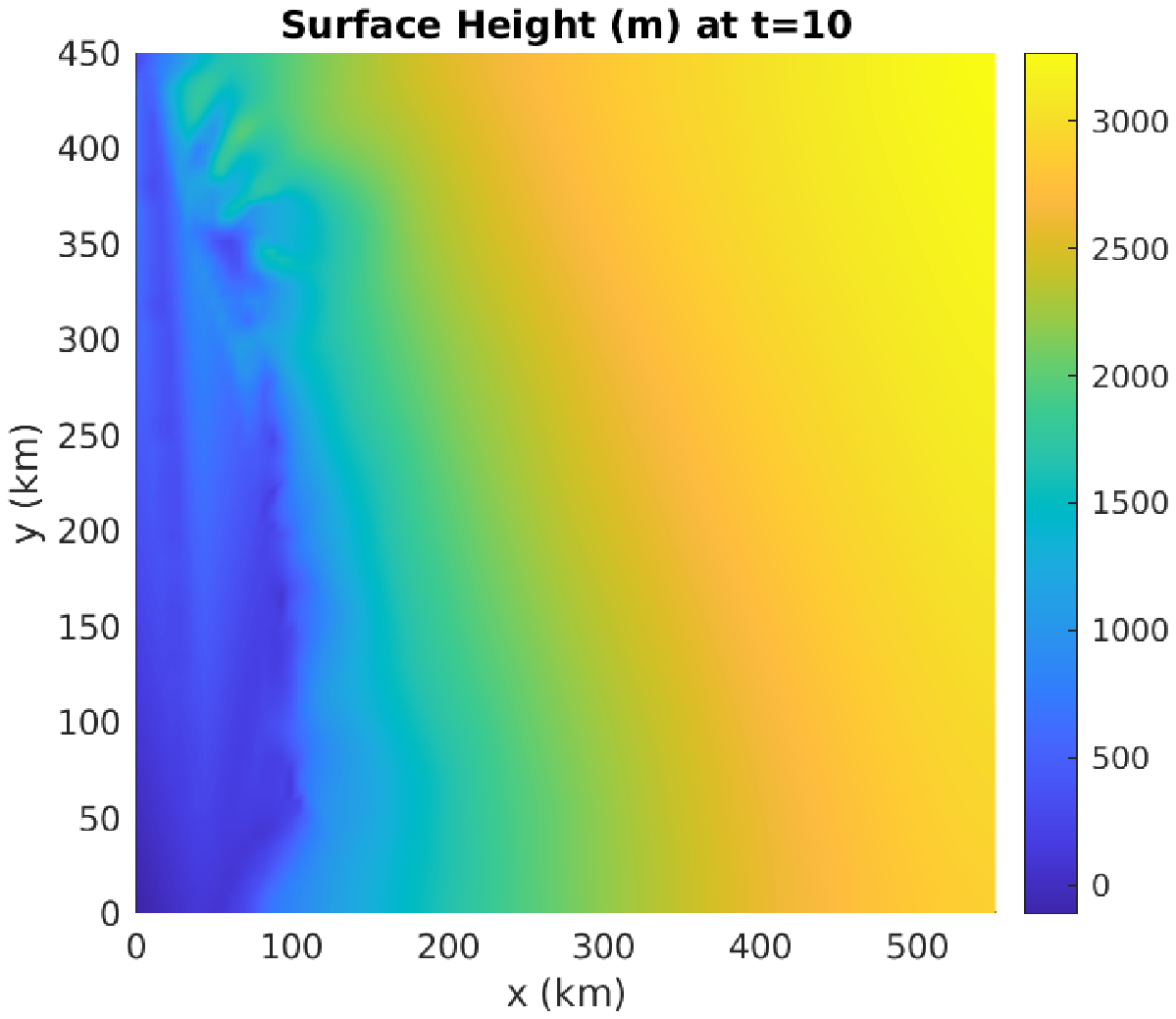}
  \caption{Top left: accumulation/ablation forcing term; top right: log basal friction; bottom: surface height at time $t=5$ (left) and $t=10$ (right). 
  }
  \label{fig:forward_model}
\end{figure}

 We generate data by solving \eqref{eqn:sia_pde_hdsa} on a $101\times 101$ mesh with $121$ time steps from $t=0$ to $t=T=10$ years and evaluating the surface velocity given in \eqref{eqn:vel}. This is a spatial resolution of $5.45$ kilometers and a time resolution of $0.083$ years. To avoid an ``inverse crime," the data is interpolated onto  a $71\times  71$ mesh with $61$ time steps, and $5\%$ Gaussian noise is added. The noise covariance is taken as $\mgamma_\text{noise} = \sigma^2_\text{noise}\mi$ with $\sigma_\text{noise} = 50$. We solve the inverse problem on a  $71 \times 71$ mesh with $61$ times steps over $10$ years.

\subsection{Results for MAP estimate}\label{ssec:inv_results}

We compute the MAP point by solving~\eqref{eqn:basal_inversion} for $\thetabar=\zero$ using a truncated conjugate gradient trust region algorithm as outlined in Lines~\ref{alg:tr_loop_start}-\ref{alg:tr_loop_end} in Algorithm~\ref{alg:full}. Exact gradients and Hessian vector products are computed using adjoint and incremental adjoint equations. The initial guess for the bedrock topography was given by a highly smoothed version of the ``true" bedrock topography.

In Table~\ref{tab:map_history} the optimization history for the computation of the MAP point is given. It terminates upon achieving a gradient norm less than $10^{-7}$. We see that the gradient norm has reduced 9 orders of magnitude and that the step size decreases as the iterations increase, indicating the convergence of the optimizer.

Visualizations of the initial guess, MAP estimate, true bedrock topography, and difference are provided in Figure~\ref{fig:map_results}.
The MAP estimate captures a majority of the features present in the true bedrock topography and the largest deviations are in the mountain in the northwest corner which is reflected in the diff plot. The bedrock is rough in that region and due to the smoothing prior, the reconstruction fails to capture some of the fine scale features. It is important to note that computing the MAP point provides an estimate of the unknown bedrock topography but it does not address how the uncertainty of the auxiliary parameters (log basal friction and forcing) will influence the uncertainty in the bedrock topography. The influence of uncertainty in the log basal friction and the forcing is explored through HDSA.

\begin{figure}[h]
\centering
\includegraphics[scale=0.49]{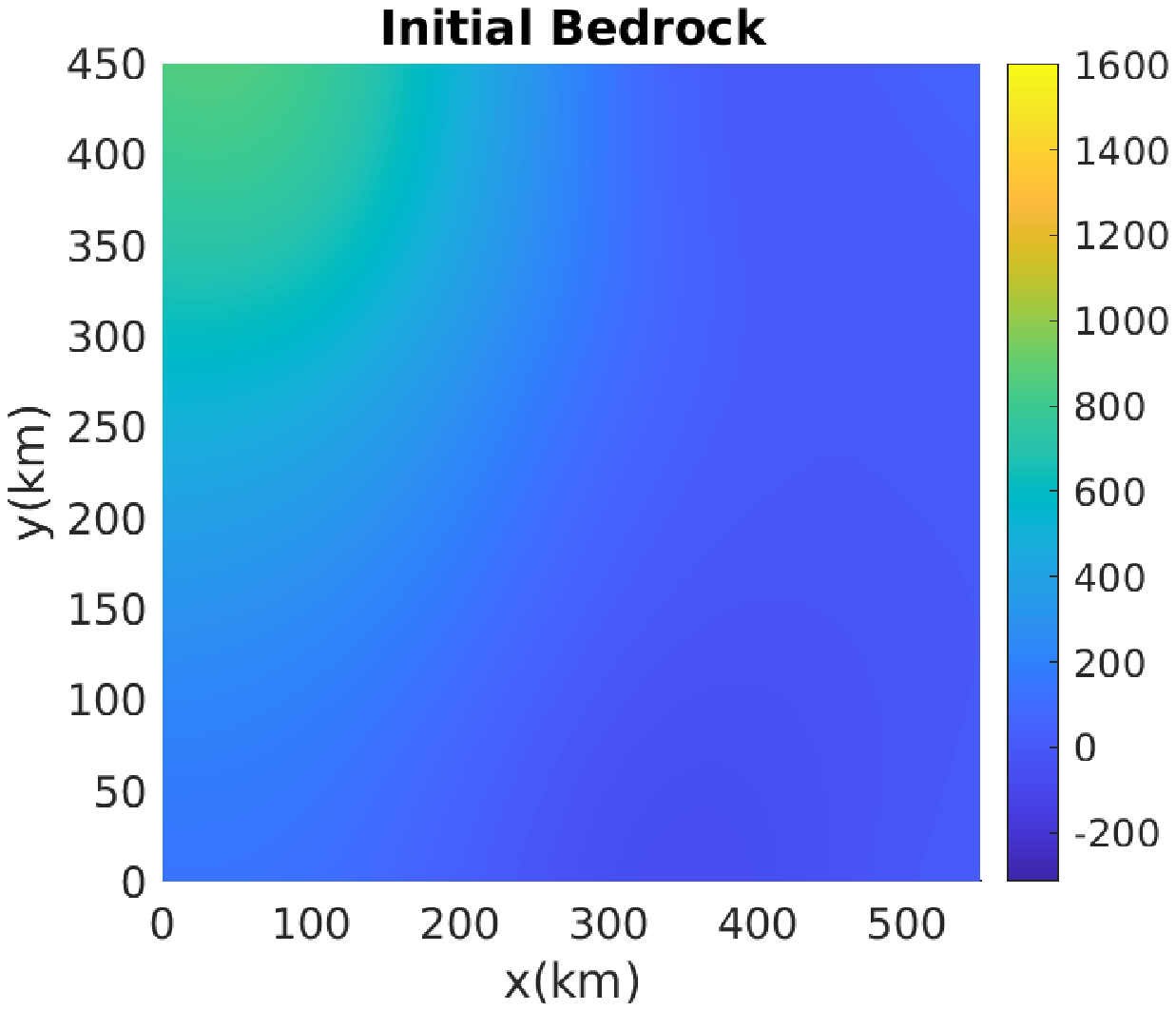}
\includegraphics[scale=0.49]{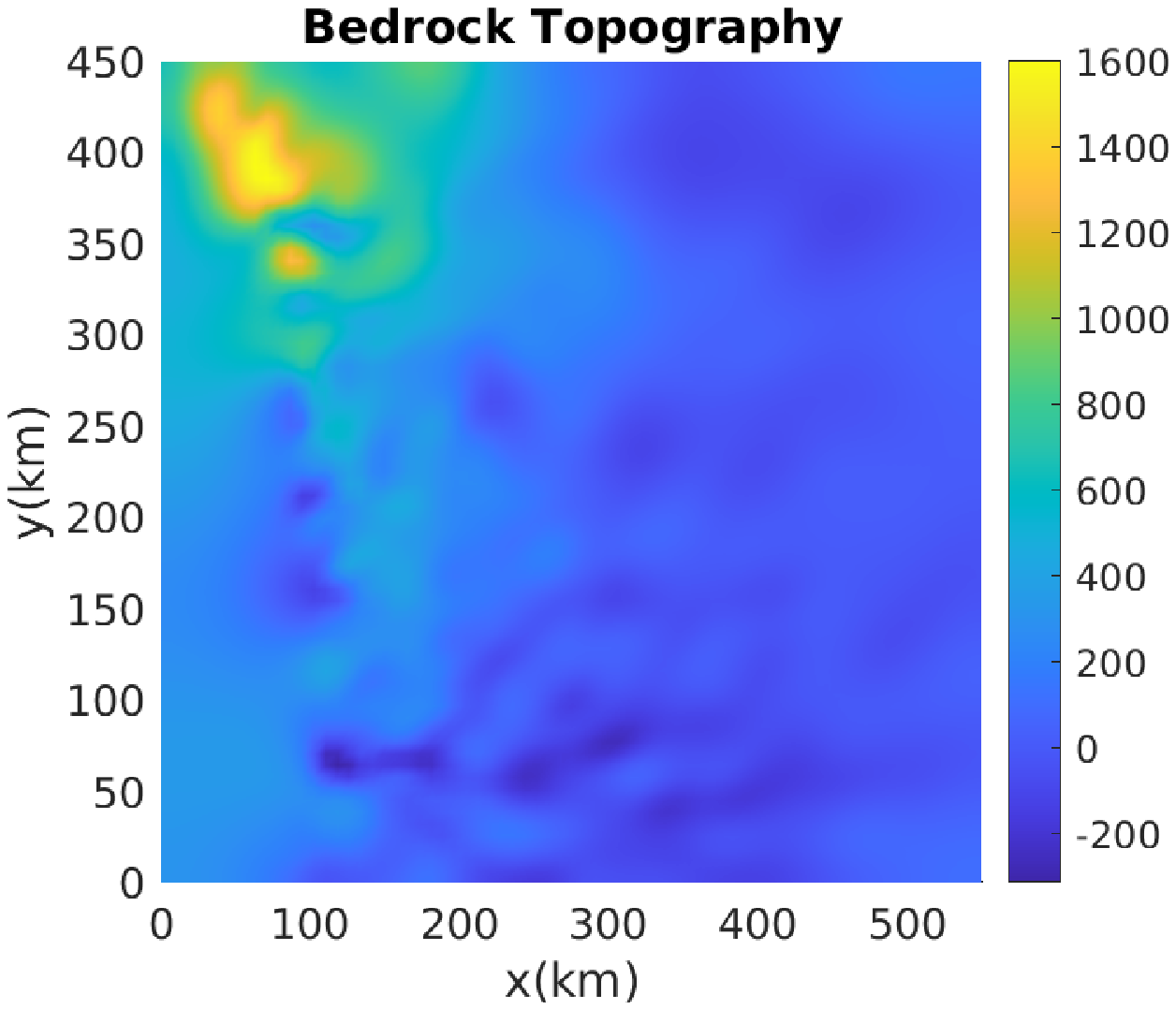} \\
\includegraphics[scale=0.49]{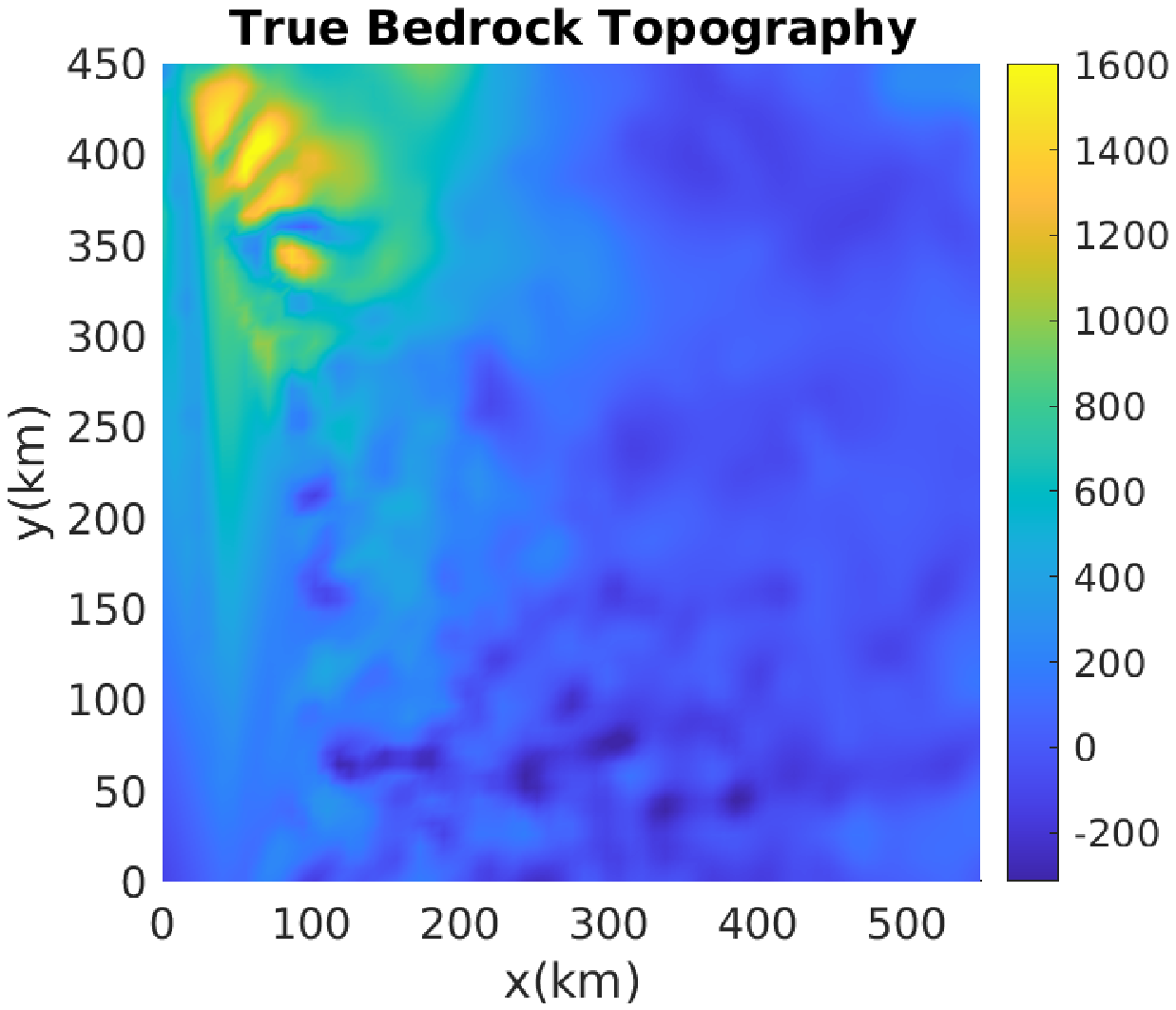} 
\includegraphics[scale=0.49]{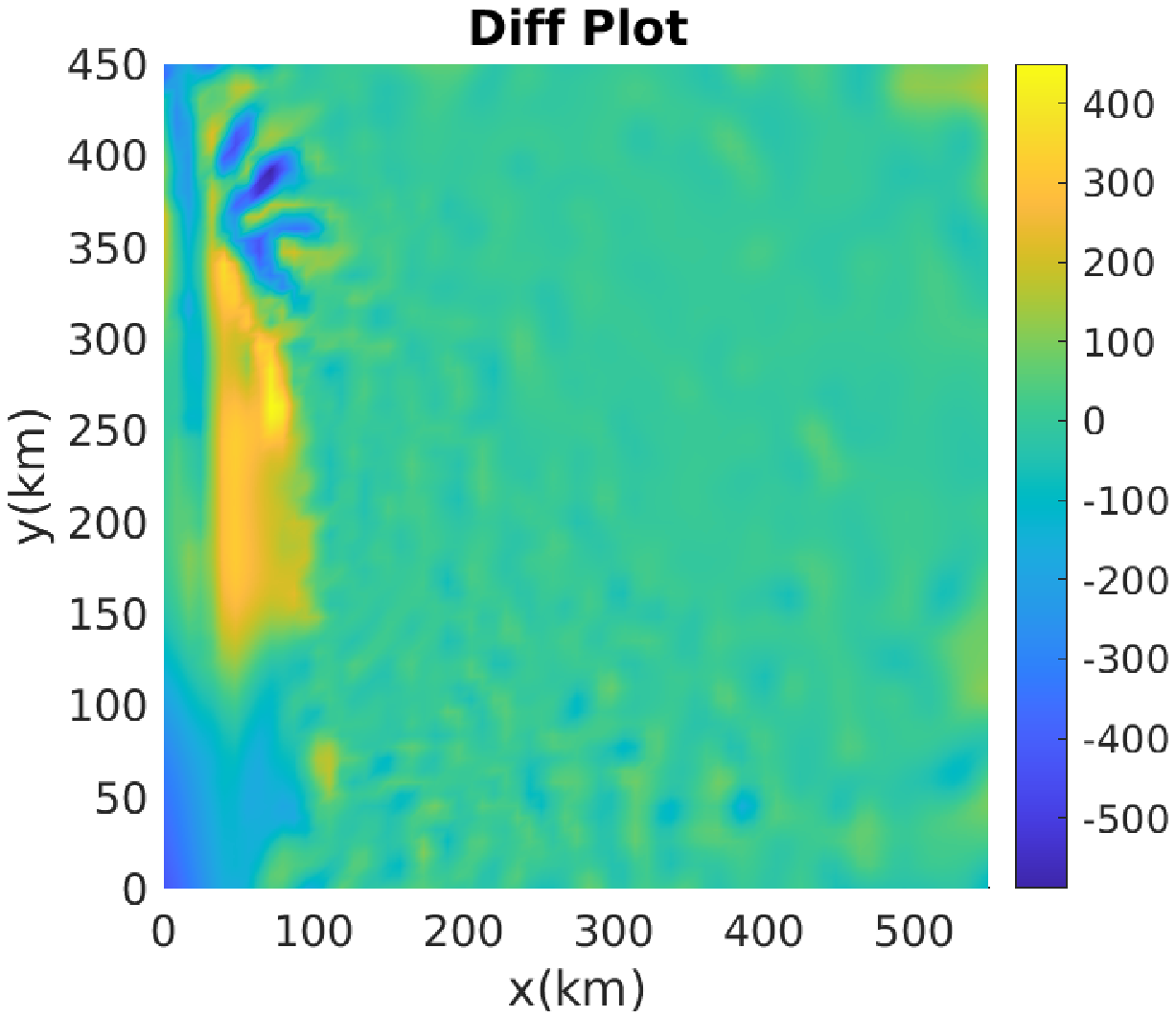}
\caption{Top left: initial bedrock topography, top right: MAP point, bottom left: true bedrock topography, bottom right: the difference between true and estimated bedrock topography.}
\label{fig:map_results}
\end{figure}

\begin{table}[h]
\centering
\begin{tabular}{c|c|c|c}
Iteration & Objective & Gradient Norm & Step Size \\
\hline 
0 & $5.83 \times 10^{4}$ & $2.66 \times 10^{1}$ & N/A \\
2 & $1.65 \times 10^{4}$ & $1.0352 \times 10^{1}$ & $2.5 \times 10^{3}$ \\
4  &   $6.36 \times 10^{3}$ &  6.089  &  $2.5 \times 10^{3}$ \\
10 &   $3.82\times 10^{3}$ & $8.55\times10^{-1}$ & $6.10\times10^{2}$\\
25   & $3.22\times10^{3}$  &  $9.72\times10^{-8}$ &   $7.36\times10^{-2}$ \\
\end{tabular}
\caption{Iteration history for calculating the MAP point}
\label{tab:map_history}
\end{table}

\subsection{Result for the LIS}

The LIS eigenpairs were calculated using Lines~\ref{alg:start_eig}-\ref{alg:end_eig} in Algorithm~\ref{alg:full} with an initial target rank $r_0 = 230$, rank increment $\Delta r = 256$, oversampling parameter $p=26$, and minimum of eigenvalue threshold $\lambda_\text{min}=0.1$. The literature recommends $p\sim 20$. but we chose $p$, $r_0$, and $\Delta_r$ to leverage parallelism with 16 compute nodes (16 cores per node). Our motivation for choosing $\lambda_\text{min}=0.1$ was to ensure that the decay of the eigenvalues was sufficient for capturing parameter directions influenced by the likelihood and prior. The decay was achieved within $2$ iterations of the for loop in the Randomized Generalized Hermitian Eigenvalue Procedure. Therefore we computed a total of $742$ eigenvalues which are shown in Figure~\ref{fig:lis_eigvals}. To calculate the sensitivity indices and approximate the posterior covariance we used all $742$ eigenvalues. In the following subsection, we demonstrate the reduction of uncertainty gained from the Gaussian approximation of the posterior by visualizing samples and variance from the prior and posterior. 

\begin{figure}[h]
\centering
\includegraphics[scale=0.6]{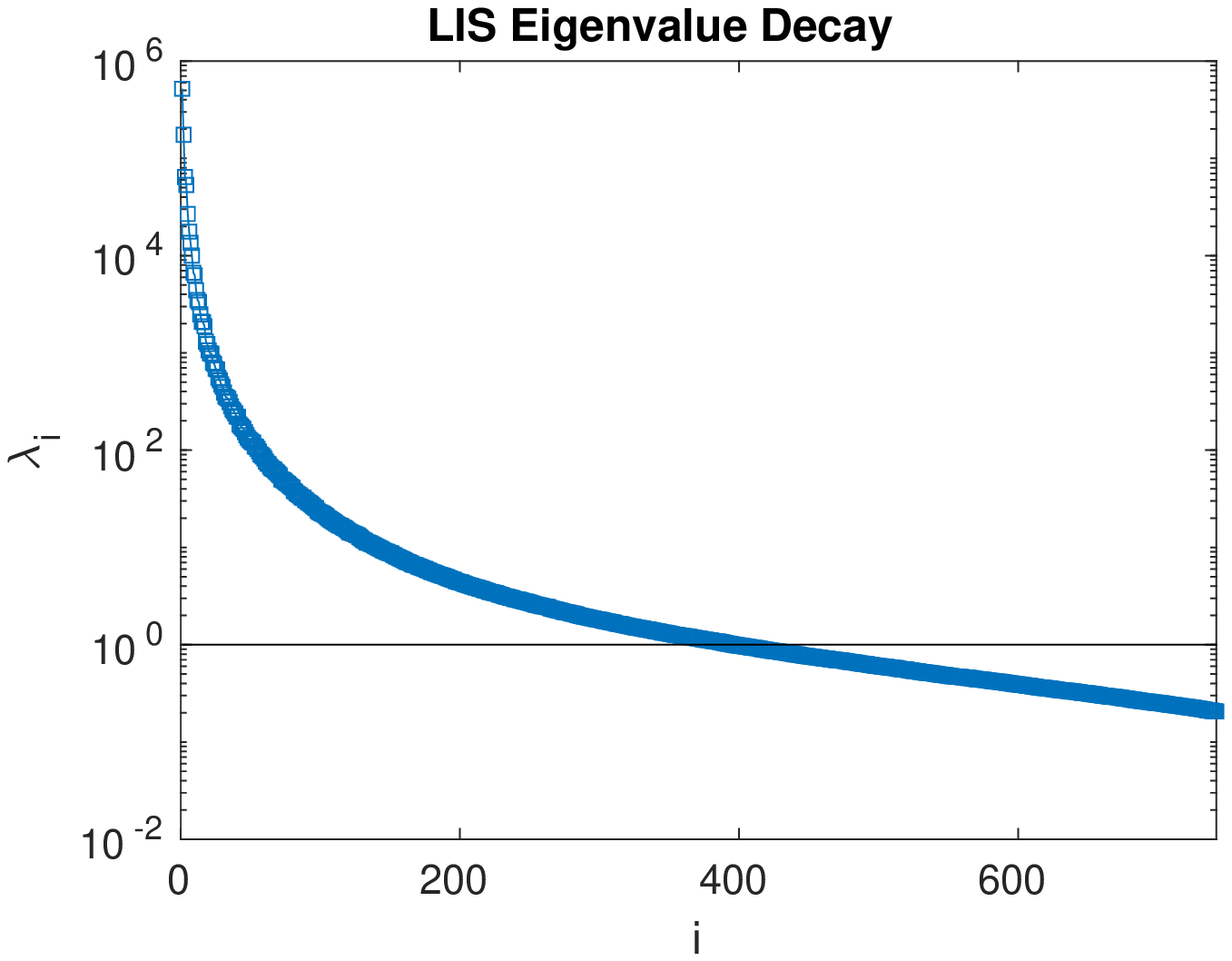}
\caption{Logarithmic plot of spectrum for~\eqref{eq:gevp}. The line $\lambda=1$ corresponds to an equal contribution from the likelihood and prior}
\label{fig:lis_eigvals}
\end{figure}


\subsection{Results for sensitivity indices}\label{ssec:hdsa_results}
Sensitivity indices for log basal friction and forcing are computed as outlined in lines~\ref{alg:b_ei}-\ref{alg:sens} of Algorithm~\ref{alg:full} and plotted in Figure~\ref{fig:sensitivities}. We make the following observations about the sensitivity indices.
\begin{itemize}
\item The magnitude of the log basal friction indices is much greater than the forcing, i.e. changes in log basal friction are more impactful on bedrock topography estimation than the forcing. This is consistent with physical intuition since the log basal friction is representative of ``slipping" between the ice sheet and the underlying bedrock which is known to have a more significant effect over time horizons of 10-20 years~\cite{lehner_2020}.

\item The interpretation of high sensitivity regions in the log basal friction is that perturbing the log basal friction in those regions will lead to the largest changes in the bedrock topography estimate. 

\item The largest log basal sensitivities correspond to the largest differences in the reconstruction of the bedrock topography which demonstrates the joint correlation between $\z$ and $\t$ (Theorem 1). 

\item In the context of joint Bayesian inversion of bedrock topography and log basal friction, the sensitivity indices correspond to correlations in the posterior. Correlations in the joint posterior distribution identify regions where inversion is difficult and may require more data collection. 

\end{itemize}
 
\begin{figure}[h]
\centering
\includegraphics[scale=0.49]{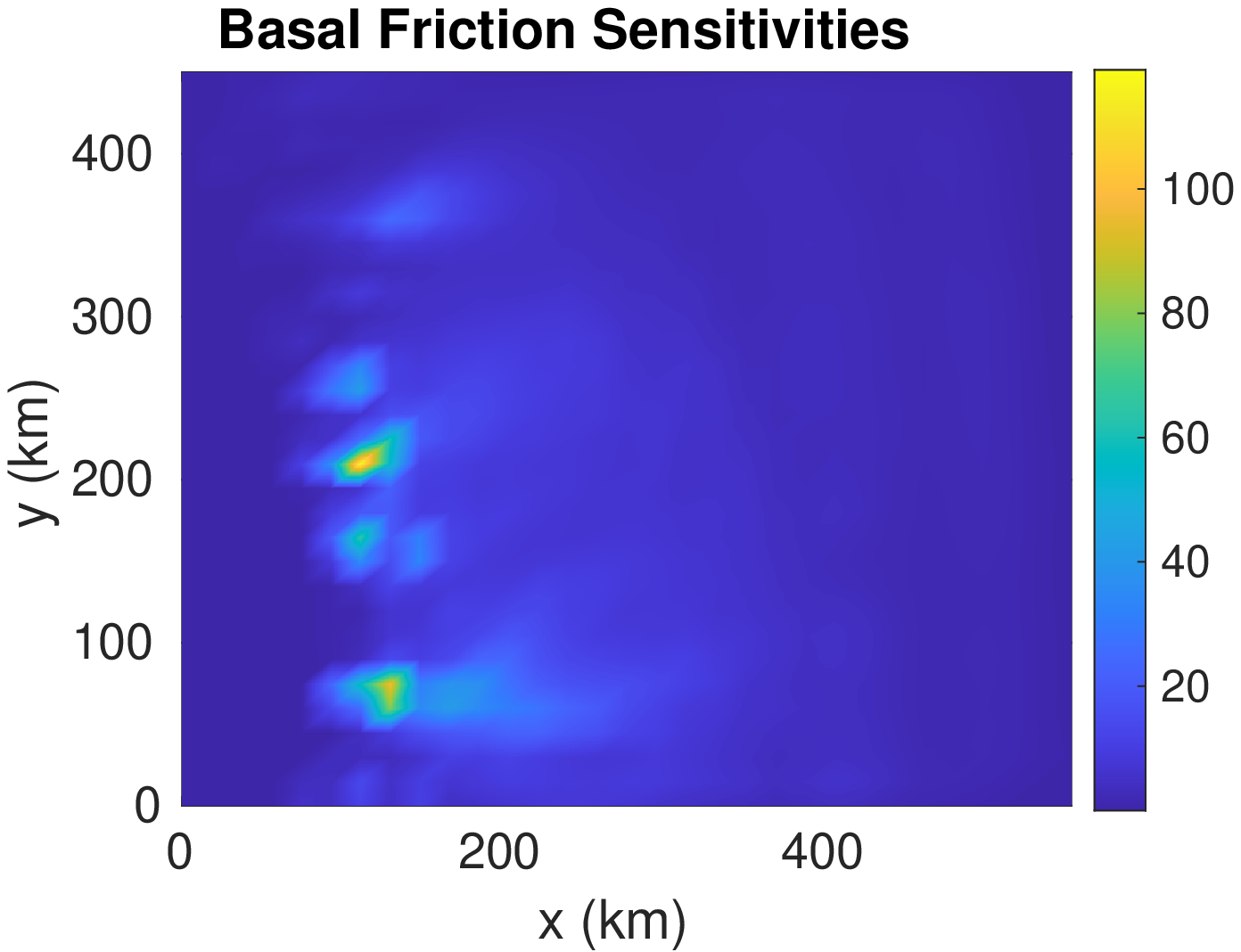}
\includegraphics[scale=0.49]{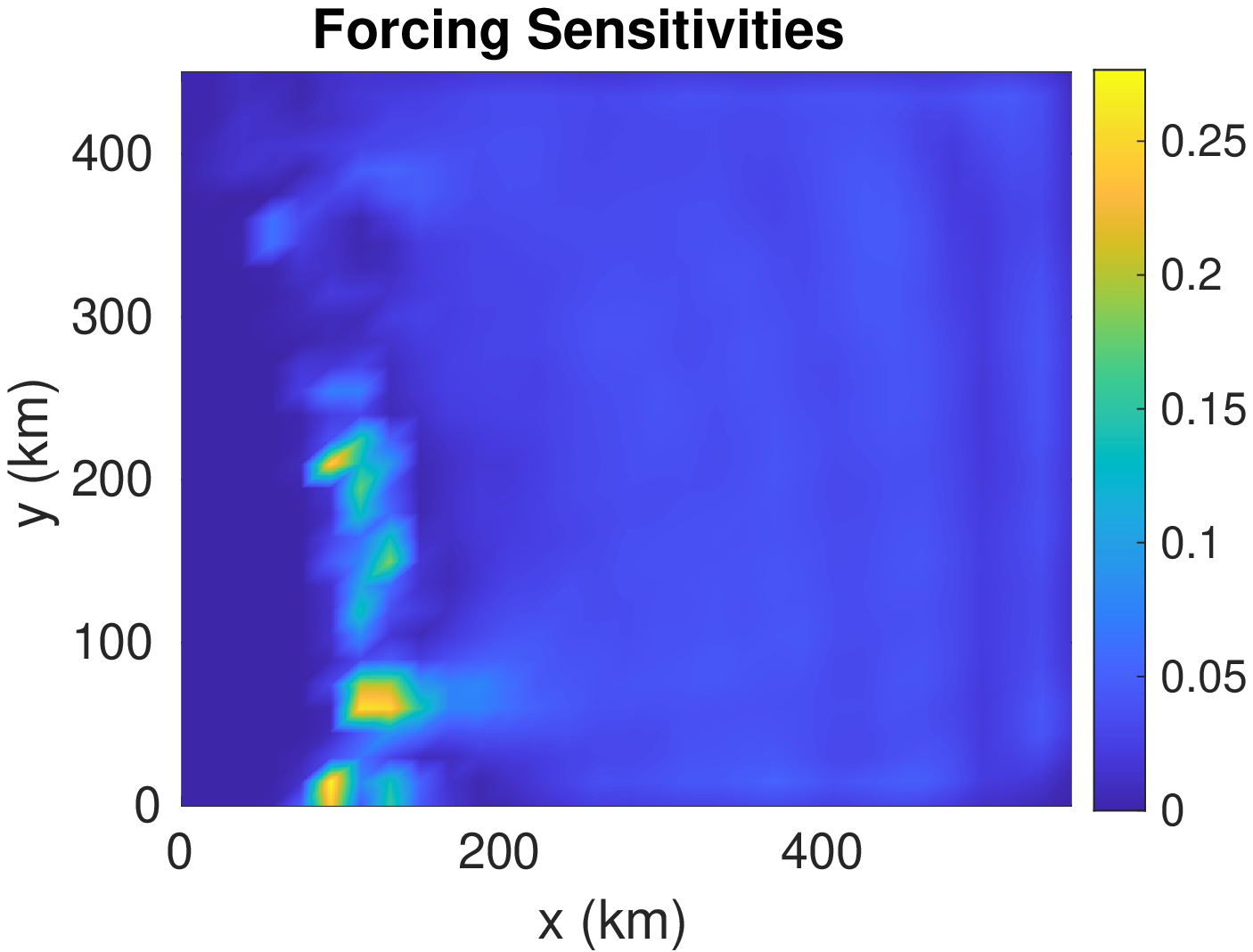}
\caption{Left: Sensitivity heatmap for the log basal friction, right: sensitivity heatmap for the forcing.}
\label{fig:sensitivities}
\end{figure}

%

\subsection{Results for Posterior samples and variance}\label{ssec:samp_var_results}
As a result of the eigenvectors calculations in the LIS, we compute approximate posterior samples and a variance estimate as described in lines~\ref{alg:laplace}-\ref{alg:var} of  Algorithm~\ref{alg:full}. We plot samples from the prior and posterior distribution in Figure~\ref{fig:samples}. The prior variance computation uses the Diag++ algorithm,~\cite[Algorithm 1]{Baston2022StochasticDE}, with parameters $s_D = 2100$ matrix-vector products. The approximation to the posterior variance is computed using the prior variance and~\eqref{eqn:postvar} with $r=742$. The approximate prior and posterior variances are plotted in Figure~\ref{fig:variance}; by contrasting the two plots one can gain insight into the reduction in uncertainty, from the prior to the posterior, obtained by solving the inverse problem. 

The posterior samples capture most of the heterogeneous structure in the domain including the mountain in the top left corner. In addition, the posterior samples have a similar spatial pattern to the MAP. The well-captured regions correspond to areas where we have more confidence in our estimate of bedrock topography.  

These observations are reflected in the plot of the posterior variance (Figure~\ref{fig:variance}, right panel). In the bottom left corner and the top left corner of the posterior variance plot, the uncertainty is much higher implying that the data is not as informative as the prior in those areas. This observation is also reflected in the inconsistent structure of the posterior samples in these regions.  The highly uncertain areas also correspond to the highest error regions in the diff plot in Figure~\ref{fig:map_results}.

\begin{figure}[ht]
\centering
\includegraphics[width=0.33\textwidth]{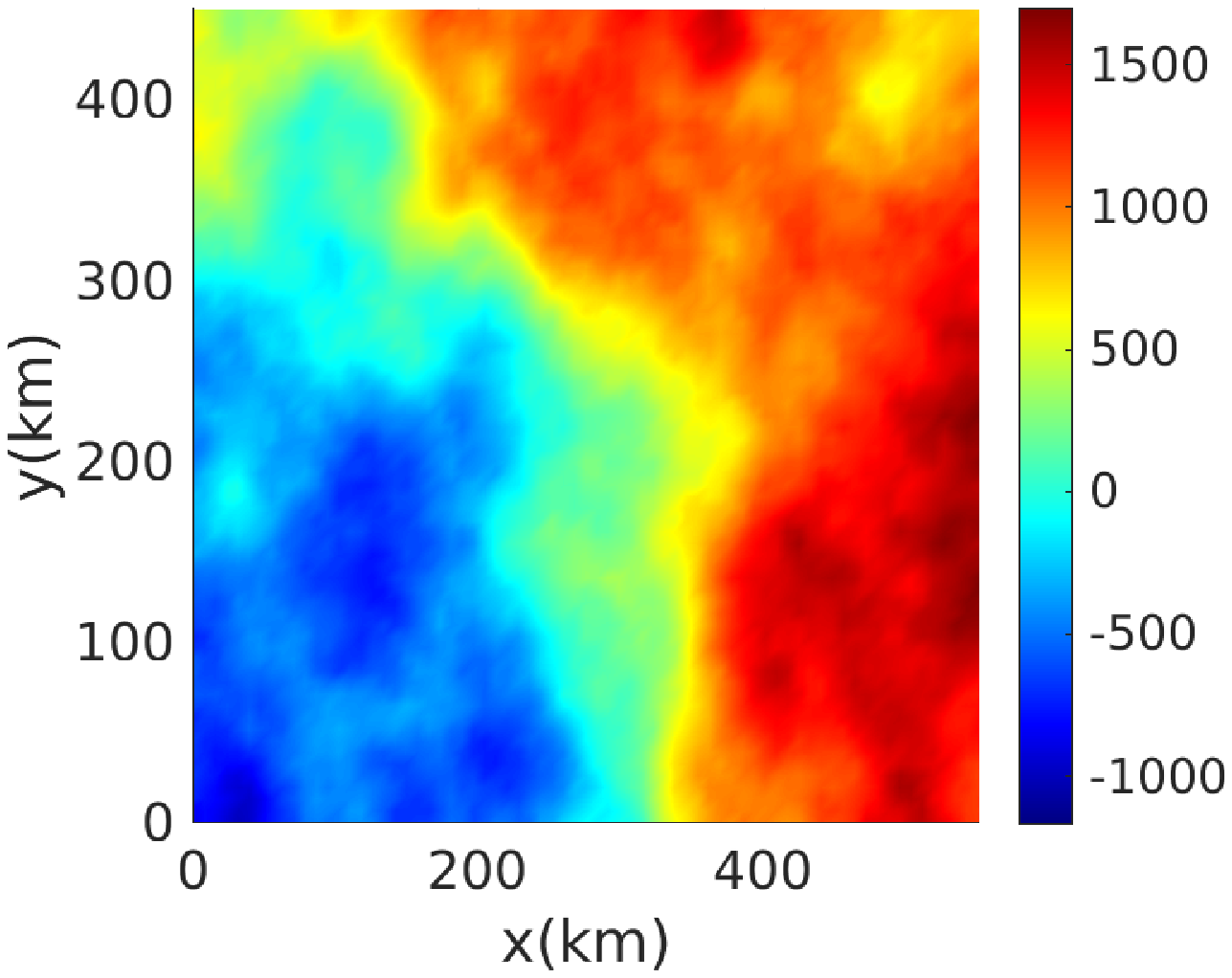}
\includegraphics[width=0.33\textwidth]{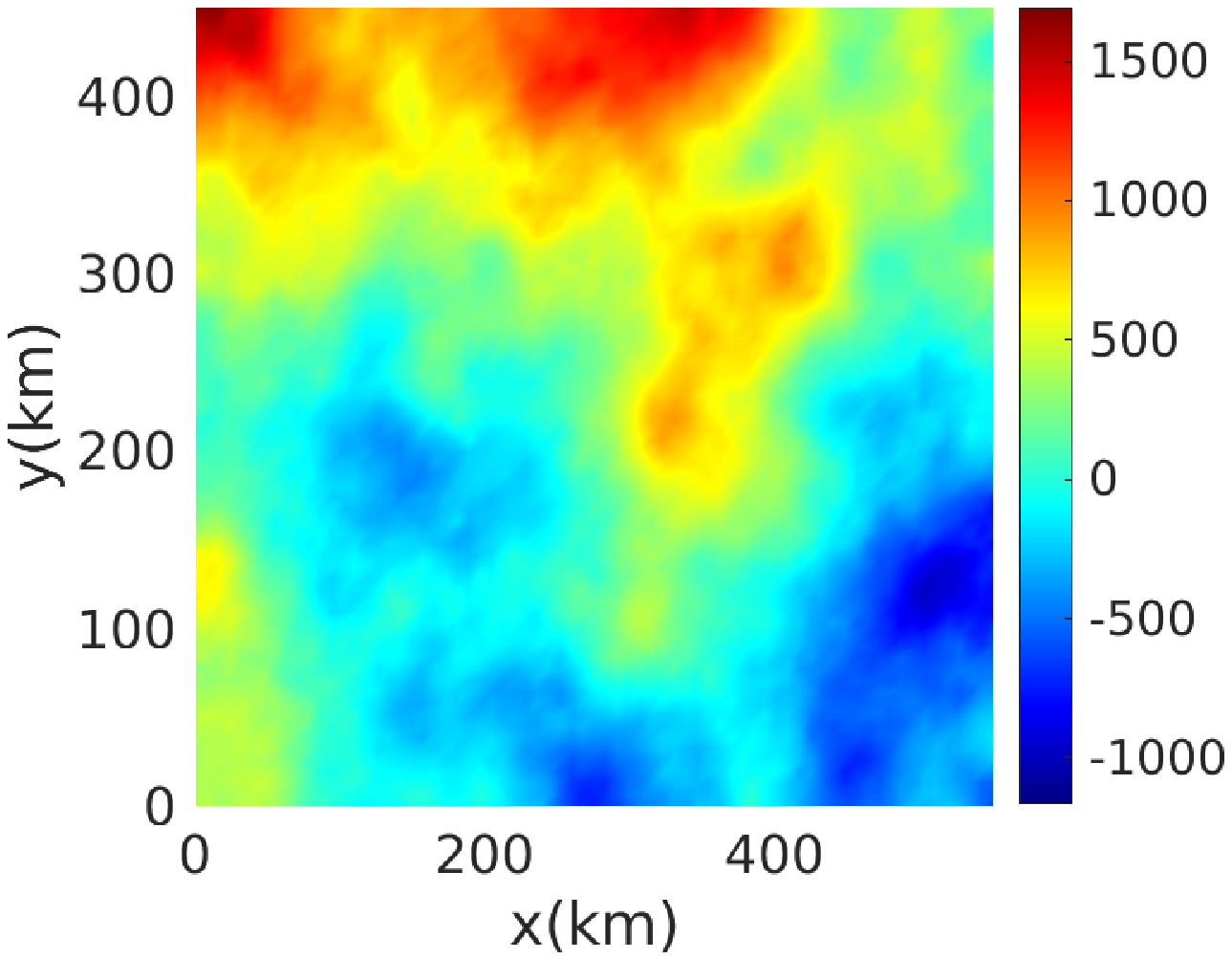}
\includegraphics[width=0.33\textwidth]{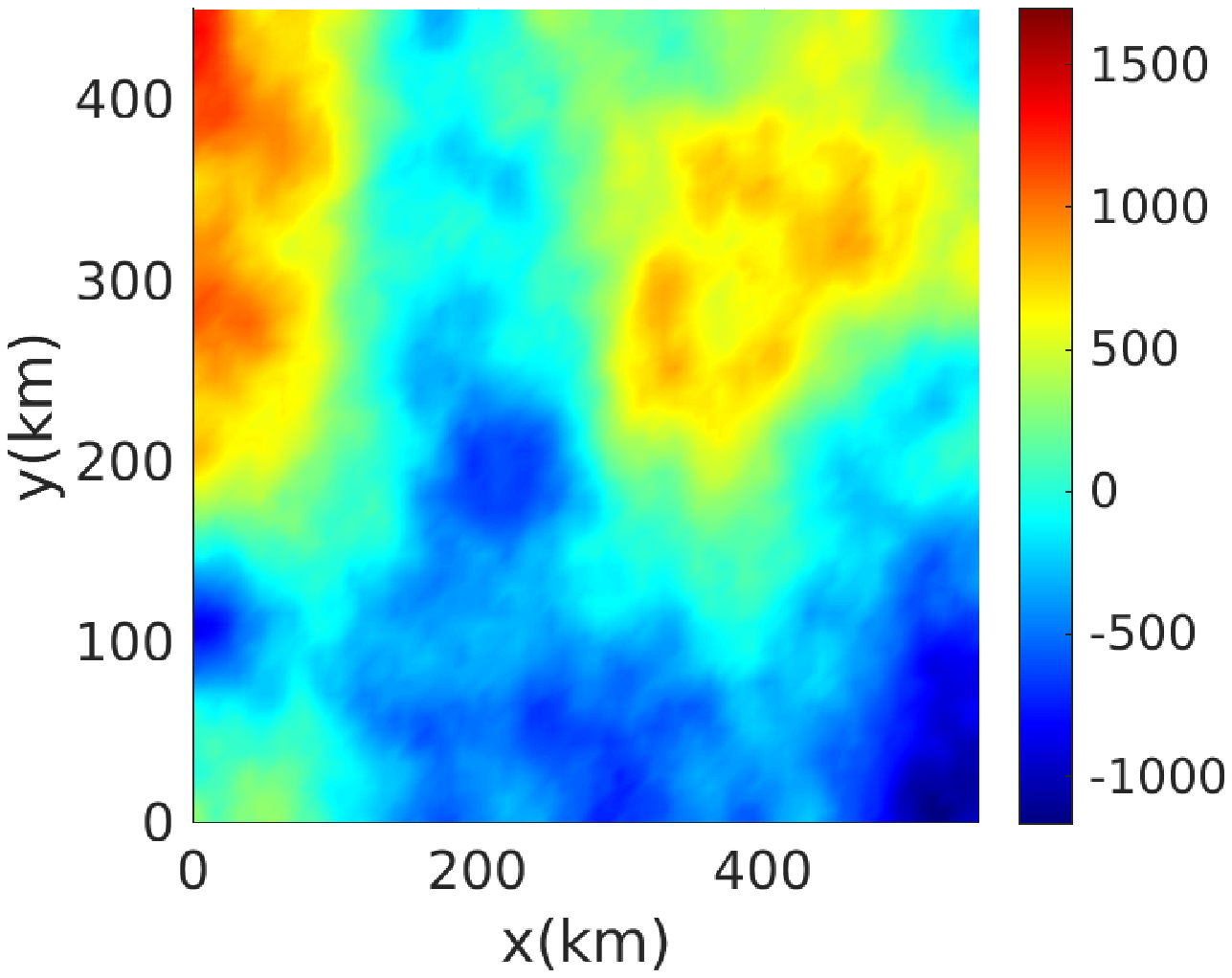} \\ 
\includegraphics[width=0.33\textwidth]{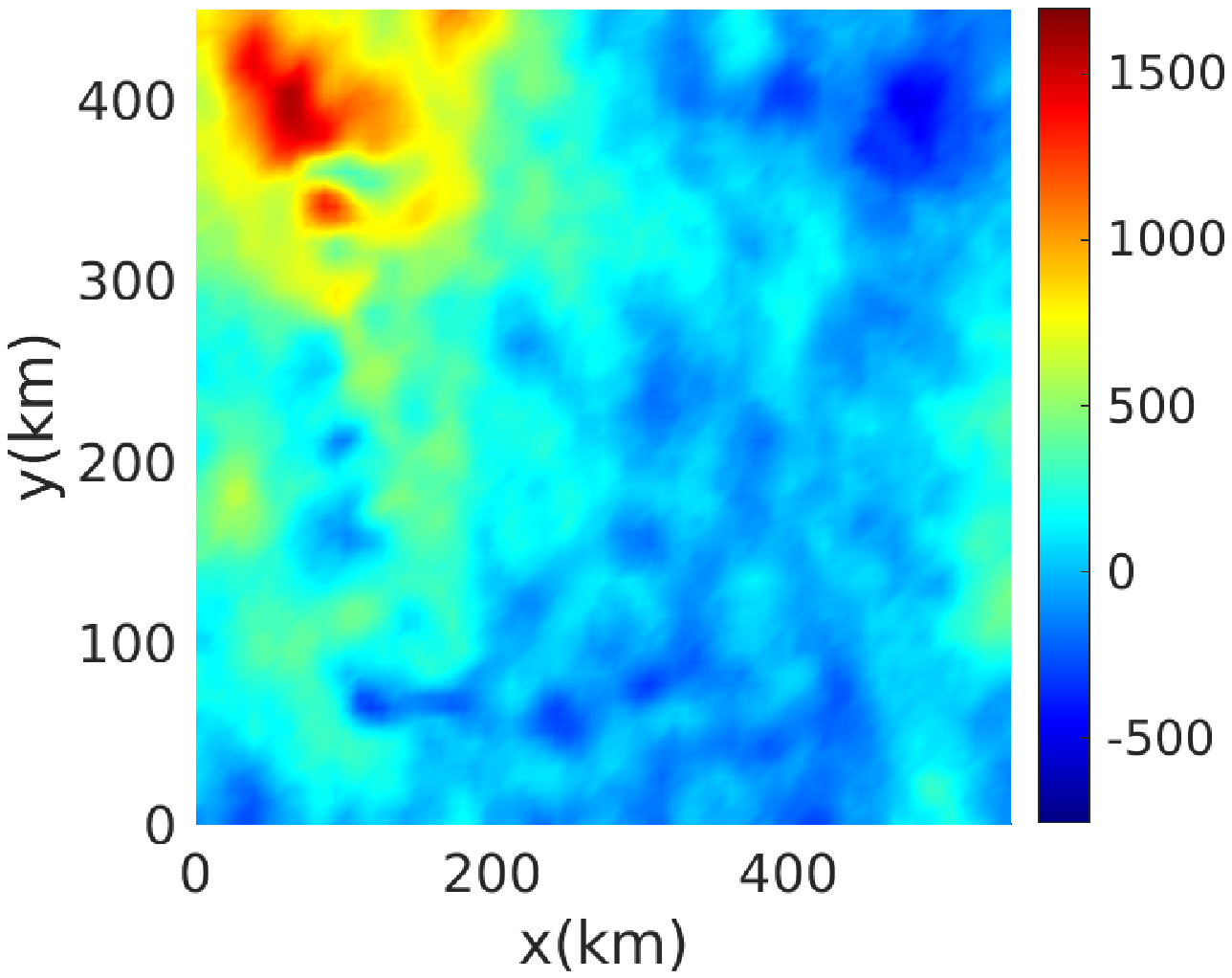}
\includegraphics[width=0.33\textwidth]{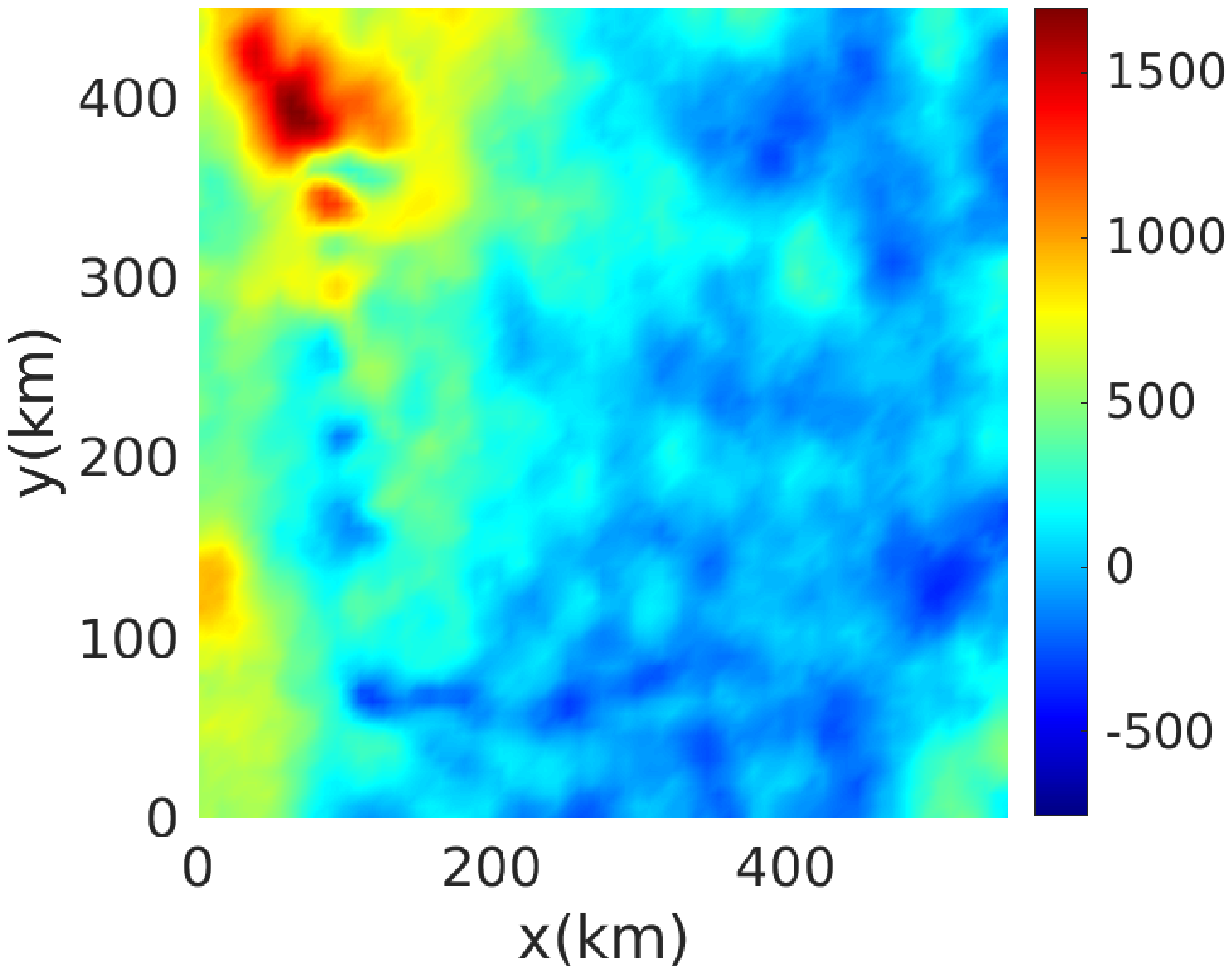}
\includegraphics[width=0.33\textwidth]{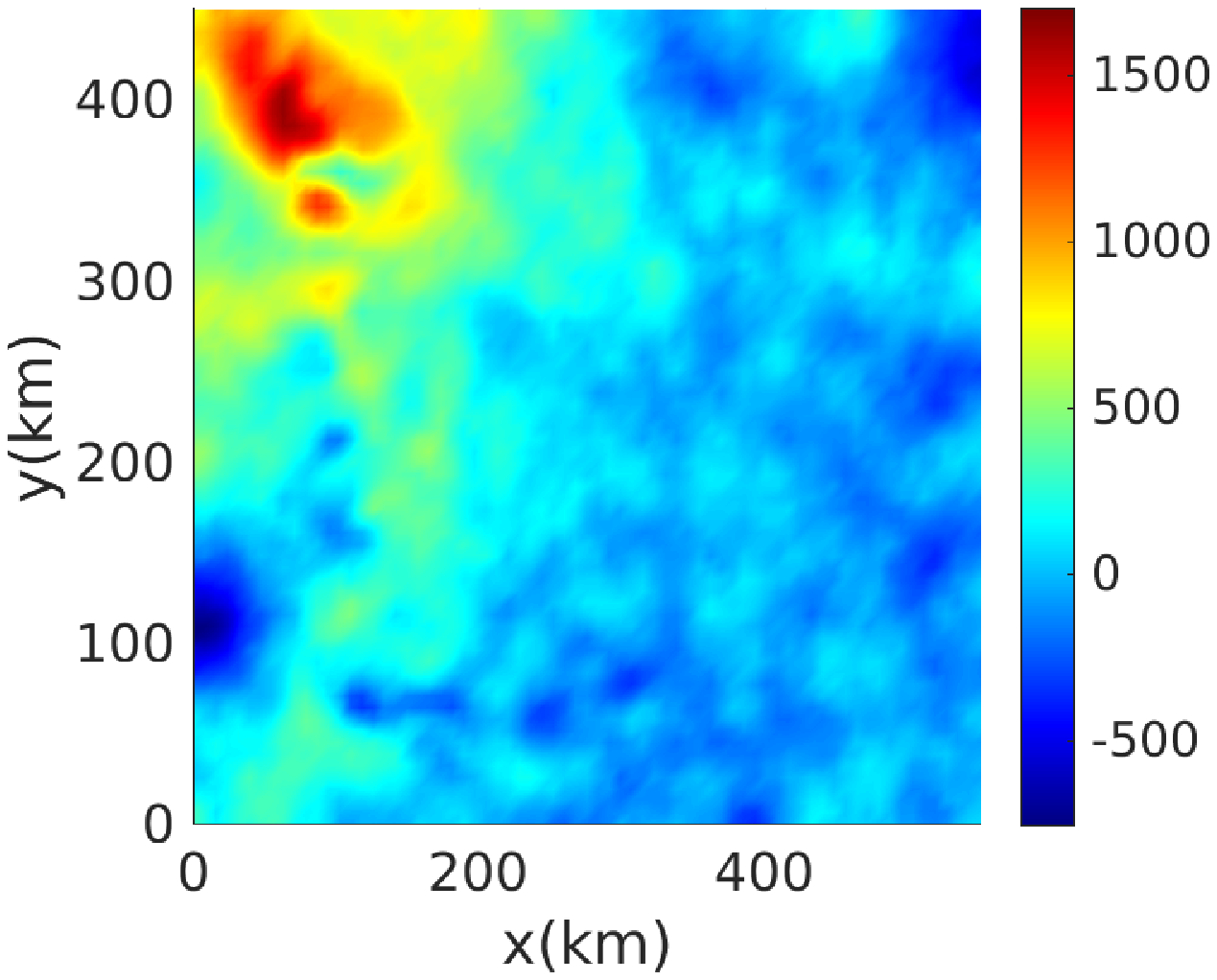}
\caption{Samples from the prior distribution, bottom: samples from the posterior distribution.}
\label{fig:samples}
\end{figure}


\begin{figure}[h]
\centering
\includegraphics[width=0.45\textwidth]{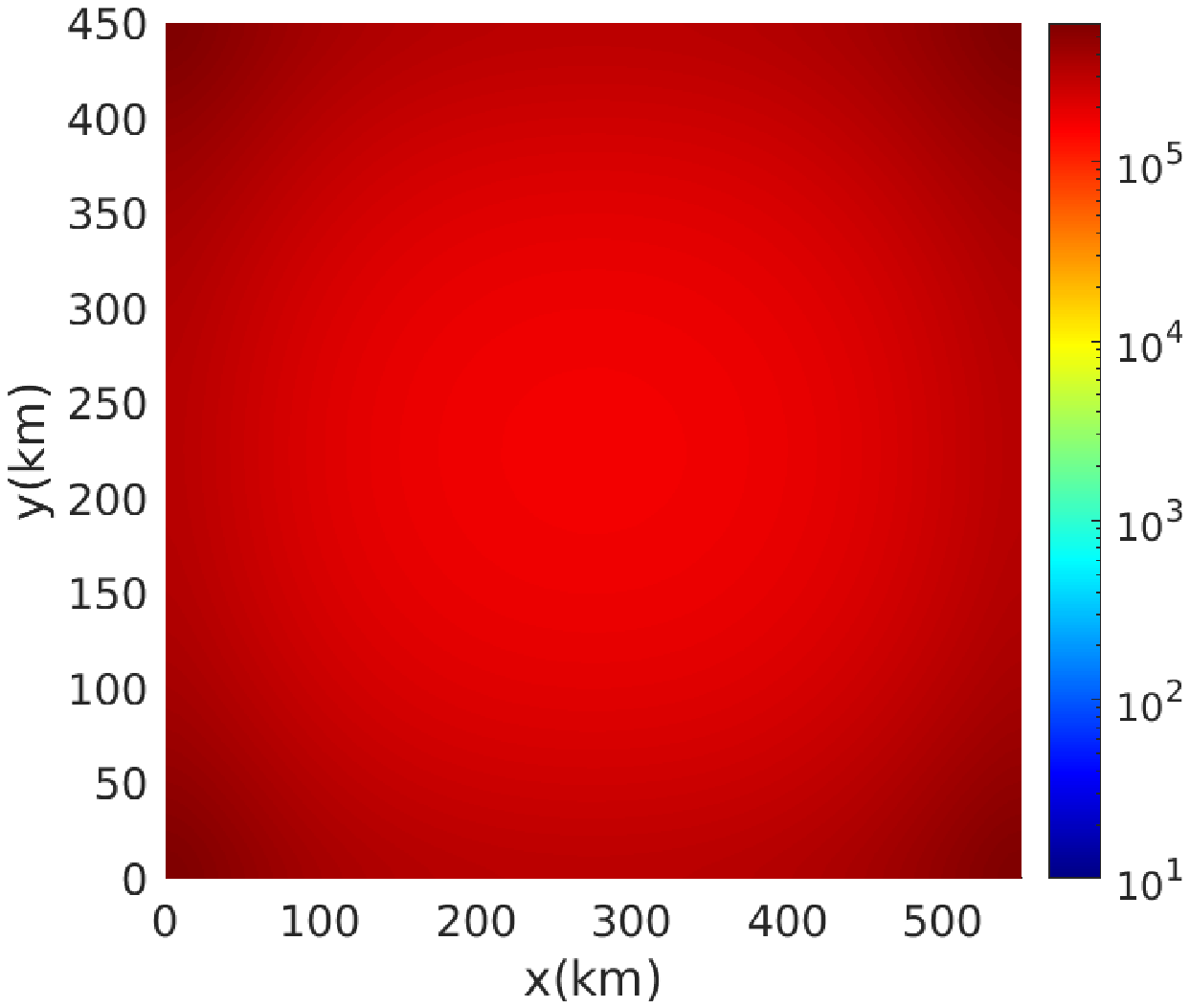}
\includegraphics[width=0.49\textwidth]{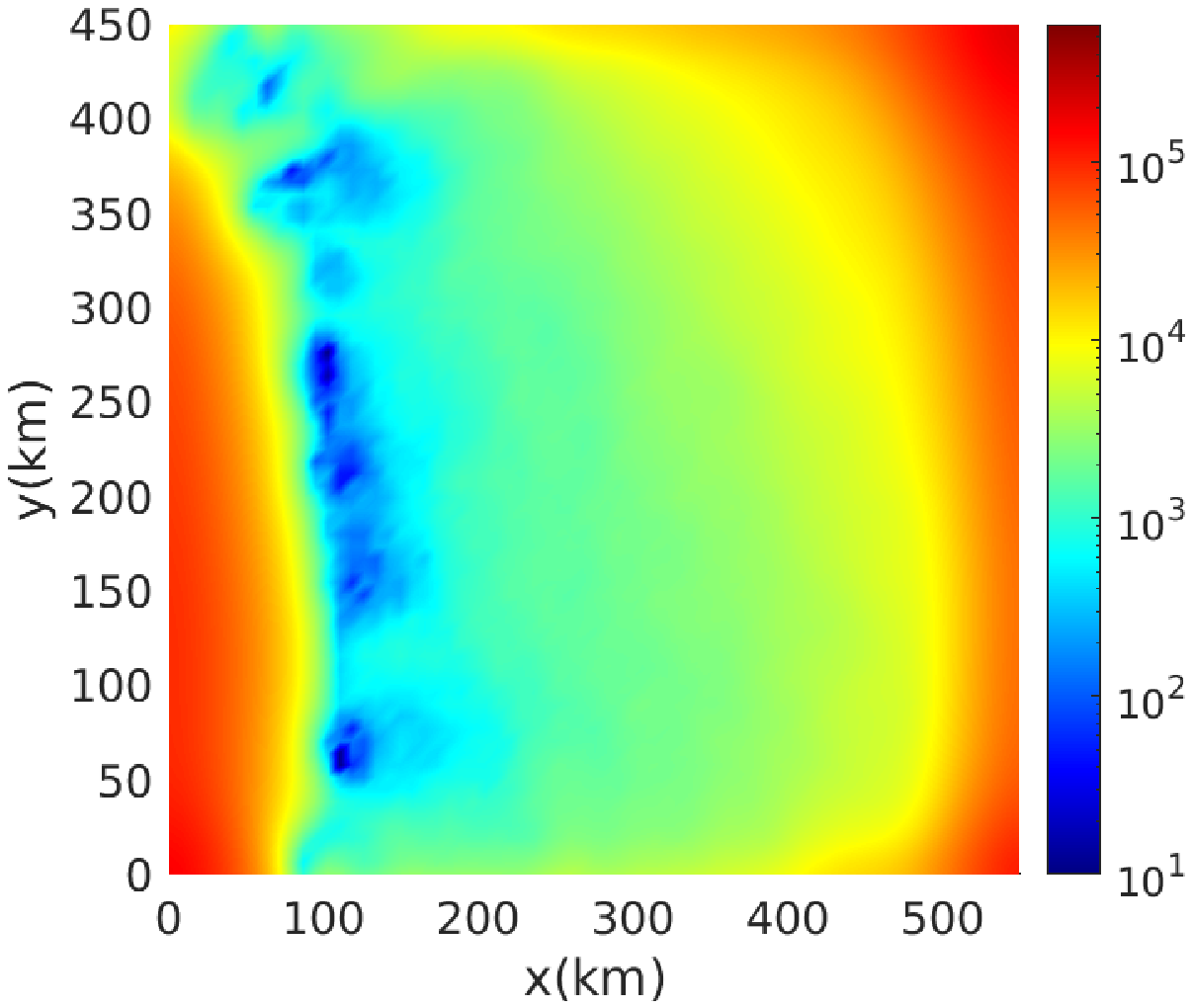}
\caption{Left: prior pointwise variance, right: posterior pointwise variance.}
\label{fig:variance}
\end{figure}


\section{Conclusions} \label{sec:conclusion}

This article presents a new algorithm that combines hyper-differential sensitivity analysis (HDSA) and approximate sampling from the posterior distribution for Bayesian inverse problems. HDSA determines the influence of uncertainties on the solution of large-scale inverse problems and is shown to have a Bayesian interpretation in terms of correlations in the joint posterior distribution. By projecting onto the likelihood informed subspace (LIS), the sensitivities are computed efficiently and samples from the Laplace approximation of the posterior distribution are computed as a by-product of the LIS computation. Motivated by large-scale problems and by the added complexity of multiple sources of uncertainty, we demonstrate our methods on a nonlinear dynamic ice sheet model whose parameters are based on realistic values from the Greenland ice sheet. As a proof of concept, we generate synthetic data to invert on bedrock topography and demonstrate HDSA with respect to basal friction and surface forcing.


We conclude that basal friction is much more important than the surface forcing in the context of bedrock topography inversion.  Although this makes intuitive sense from a physics perspective, HDSA quantifies the differences exactly and provides a spatial characterization of sensitivities for basal friction and surface forcing. From a Bayesian perspective, this implies that the bedrock topography and basal friction are highly correlated in their joint posterior. With these conclusions in hand, one may consider the implementation of higher-fidelity basal friction modeling or re-prioritization the data acquisition program. 

HDSA insight is enabled through several computational foundations.  Adjoint-based derivatives are used to calculate gradients and Hessians as part of a trust-region Newton-CG optimization algorithm, which is predicated on parallel numerical linear algebra. A randomized Generalized Hermitian Eigenvalue solver is used to efficiently compute the LIS projector asynchronously.   

Inversion using observed data (instead of synthetically generated data) can create solver convergence issues making the interpretation of HDSA questionable. However projecting the sensitivities onto the LIS subspace provides a clear interpretation as discussed in~\cite{hdsa_ill_posed_inv_prob}.

\section*{Acknowledgements}
This paper describes objective technical
results and analysis. Any subjective views or opinions that might be
expressed in the paper do not necessarily represent the views of the
U.S. Department of Energy or the United States Government. Sandia
National Laboratories is a multimission laboratory managed and
operated by National Technology and Engineering Solutions of Sandia
LLC, a wholly owned subsidiary of Honeywell International, Inc., for
the U.S. Department of Energy's National Nuclear Security
Administration under contract DE-NA-0003525. SAND2022-16621 O.

\bibliographystyle{abbrv}
\bibliography{dasco}

\end{document}